\documentclass[pra,showpacs,graphics,twocolumn,floatfix,mathbbm,a4paper,nofootinbib]{revtex4-1}

\usepackage{amsmath}
\usepackage{amsthm}
\usepackage{latexsym}
\usepackage{amsfonts}
\usepackage{amssymb}
\usepackage{bbm,dsfont}
\usepackage{color}
\usepackage{graphicx}
\usepackage{enumerate}
\usepackage{subfigure}

\usepackage[normalem]{ulem}
\usepackage{mathrsfs}
\usepackage{mathtools}

\newtheorem{proposition}{Proposition}
\newtheorem{theorem}{Theorem}
\newtheorem{lemma}{Lemma}
\newtheorem{corollary}{Corollary}

\theoremstyle{definition}

\newcommand{\e}{{\rm e}}

\definecolor{darkgreen}{rgb}{0,0.4,0.3}

\newcommand{\R}{\mathbb{R}} %real
\newcommand{\real}{\mathbb R} %real
\newcommand{\complex}{\mathbb C}

\newcommand{\Z}{\mathbb Z} %integer
\newcommand{\half}{\tfrac{1}{2}} %half
\newcommand{\mo}[1]{\left| #1 \right|} 

\newcommand{\abs}{\mo} %modulus

\newcommand{\hi}{\mathcal{H}} %Hilbert space
\newcommand{\lh}{\mathcal{L(H)}} %bounded linear operators on H
\newcommand{\elle}[1]{\mathcal{L} \left( #1 \right)} %bounded linear operators on #1
\newcommand{\spanno}[1]{\mathrm{span}(#1)}
\newcommand{\ip}[2]{\left\langle\,#1\,|\,#2\,\right\rangle} %inner product
\newcommand{\kb}[2]{|#1\rangle\langle#2|} %ketbra
\newcommand{\no}[1]{\left\|#1\right\|} %norm
\newcommand{\tr}[1]{{\rm tr}\left[#1\right]} %trace
\newcommand{\ran}[1]{\textrm{ran}\,#1} %range
\newcommand{\id}{\mathbbm{1}} %identity operator
\newcommand{\rank}[1]{{\rm rank}\,#1} %rank

\newcommand{\vau}{\hat{a}} % a unit 
\newcommand{\vbu}{\hat{b}} % b unit
\newcommand{\veu}{\hat{e}} % e unit
\newcommand{\vsigma}{\vec{\sigma}}% sigma

\newcommand{\C}{\mathsf{C}}%generic observable
\newcommand{\M}{\mathsf{M}}%generic observable
\newcommand{\N}{\mathsf{N}}%generic observable
\newcommand{\nN}{\widetilde{\mathsf{N}}}%noisy N
\newcommand{\jmd}{\mathsf{j}} %joint measurability degree
\newcommand{\jmdu}{\mathsf{j}_{u}} %uniform joint measurability degree

\newcommand{\en}{\mathcal{E}} %ensemble
\newcommand{\enf}{\mathcal{F}} %ensemble

\newcommand{\Zb}{\mathbb Z}
\newcommand{\ca}[1]{\mathcal{#1}}
\renewcommand{\ker}[1]{\textrm{ker}\,#1} %kernel

\newcommand{\Pg}{P_{{\rm guess}}}
\newcommand{\Ppg}{\Pg^{{\rm post}}}
\newcommand{\Prg}{\Pg^{{\rm prior}}}
\newcommand{\Pscr}{\mathscr{P}}
\newcommand{\phii}{\varphi}
\newcommand{\lam}{\lambda}
\newcommand{\vxu}{\hat{x}} % x unit
\renewcommand{\i}{{\rm i}}

\begin{document}%\setlength{\arraycolsep}{2pt}

\title[]{State discrimination with post-measurement information and incompatibility of quantum measurements}

\author{Claudio Carmeli}
\email{claudio.carmeli@gmail.com}
\affiliation{DIME, Universit\`a di Genova, Via Magliotto 2, I-17100 Savona, Italy}

\author{Teiko Heinosaari}
\email{teiko.heinosaari@utu.fi}
\affiliation{QTF Centre of Excellence, Turku Centre for Quantum Physics, Department of Physics and Astronomy, University of Turku, FI-20014 Turku, Finland}

\author{Alessandro Toigo}
\email{alessandro.toigo@polimi.it}
\affiliation{Dipartimento di Matematica, Politecnico di Milano, Piazza Leonardo da Vinci 32, I-20133 Milano, Italy}
\affiliation{I.N.F.N., Sezione di Milano, Via Celoria 16, I-20133 Milano, Italy}

\begin{abstract}
We discuss the following variant of the standard minimum error state discrimination problem: Alice picks the state she sends to Bob among one of several disjoint state ensembles, and she communicates him the chosen ensemble only at a later time. Two different scenarios then arise: either Bob is allowed to arrange his measurement set-up after Alice has announced him the chosen ensemble, or he is forced to perform the measurement before of Alice's announcement. In the latter case, he can only post-process his measurement outcome when Alice's extra information becomes available. We compare the optimal guessing probabilities in the two scenarios, and we prove that they are the same if and only if there exist compatible optimal measurements for all of Alice's state ensembles. When this is the case, post-processing any of the corresponding joint measurements is Bob's optimal strategy in the post-measurement information scenario. Furthermore, we establish a connection between discrimination with post-measurement information and the standard state discrimination. By means of this connection and exploiting the presence of symmetries, we are able to compute the various guessing probabilities in many concrete examples.
\end{abstract}

\maketitle

%%%%%%%%%%%%%%%%%%%%%%
\section{Introduction}\label{sec:intro}
%%%%%%%%%%%%%%%%%%%%%%

Quantum state discrimination is one of the fundamental tasks in quantum information processing.
In the setting of state discrimination, a quantum system is prepared in one out of a finite collection of possible states, chosen with a certain apriori probability.
The aim is then to identify the correct state by making a single measurement, assuming that the possible states and their apriori probabilities are known before the measurement is chosen. 
This can be also seen as a task of retrieving classical information that has been encoded in quantum states. 
A collection of orthogonal pure states, or more generally mixed states with disjoint supports, can be perfectly discriminated, while in other cases one has to accept either error or inconclusive result. 
These alternatives lead to two main branches of discrimination problems, called minimum error discrimination and unambiguous discrimination, respectively, that have both been investigated extensively; 
thorough reviews are provided in \cite{Sedlak09, Bergou10, Bae13}. 

Several types of variants of state discrimination problem have been introduced and studied in the literature. 
A specific variant of minimum error state discrimination problem, called state discrimination with post-measurement information, was elaborated in \cite{BaWeWi08}. 
In this task, Alice encodes classical information in quantum states and Bob then performs a measurement to guess the correct state, but Alice announces some partial information on her encoding before Bob must make his guess. 
This task was further studied in \cite{GoWe10}, and it was suggested that the usefulness of post-measurement information distinguishes the quantum from the classical world.

In this work we reveal a link between the task of state discrimination with post-measurement information and the incompatibility of quantum measurements. 
We will first formalize discrimination tasks with pre-measurement and post-measurement information in a consistent way, allowing us to compare the optimal guessing probabilities in these two cases.
(For clarification, we point out that our formulation slightly differs from the one presented in \cite{BaWeWi08} and \cite{GoWe10}. The formulations are compared in Sec.~\ref{sec:comparison}.)
We then show that pre-measurement information is strictly more favorable than post-measurement information if and only if the optimal measurements for the subensembles are incompatible.
Since incompatibility is a genuine non-classical feature \cite{HeMiZi16,Plavala16, Jencova17}, this result uncovers a peculiarity that differentiates quantum from classical measurements. 

As a technical method to calculate the optimal guessing probabilities and optimal measurements, we show that it is always possible to transform the problem of state discrimination with post-measurement information into a usual minimum error state discrimination problem; more precisely, any state discrimination problem with post-measurement information is associated with standard state discrimination for a specific auxiliary state ensemble, in such a way that the two discrimination tasks have the same optimal measurements, and the respective success probabilities are related by a simple equation.
In this way, the known results for the usual minimum error discrimination can be used for state discrimination with post-measurement information.
 
Finally, we discuss the connection between state discrimination with post-measurement information and approximate joint measurements in the cases when the optimal measurements for the subensemble discrimination problems are incompatible.
We provide several examples, showing that the approximate joint measurement is sometimes optimal, although not always.
In particular, we present an analytic solution for the problem of state discrimination with post-measurement information of two Fourier conjugate mutually unbiased bases in arbitrary finite dimension.

\emph{Notations.} We deal with quantum systems associated with a finite dimensional Hilbert space $\hi$. We denote by $\lh$ the set of all linear operators on $\hi$, and $\id\in\lh$ is the identity operator of $\hi$. The \emph{states} of the system are all positive trace one operators in $\lh$. A \emph{measurement} with outcomes in a finite set $X$ is any positive operator valued measure (POVM) based on $X$, i.e., any mapping $\M:X\to\lh$ such that $\M(x)\geq 0$ for all $x\in X$ and $\sum_{x\in X}\M(x) = \id$.

%%%%%%%%%%%%%%%%%%%%%%
\section{State discrimination with post-measurement information}\label{sec:start}
%%%%%%%%%%%%%%%%%%%%%%

%%%%%%%%%%%%%%%%%%%%%%
\subsection{General scenario}\label{sec:scenario}
%%%%%%%%%%%%%%%%%%%%%%

A \emph{state ensemble} $\en$ is sequence of states $(\varrho_x)_{x\in X}$, labeled with a finite set $X$, together with an assignment of some prior probability $p(x)$ to each label $x\in X$. 
It is convenient to regard $\en$ as a map $X\to\lh$, given by $\en(x)=p(x)\varrho_x$. 
We say that a quantum system is prepared or chosen from the state ensemble $\en$ when a label $x\in X$ is picked according to the probability distribution $p$, and the system is then set in the corresponding state $\varrho_x$.

In the standard minimum error state discrimination scenario (see Fig.~\ref{fig:usual}), there is a state ensemble $\en$ that is known to two parties, Alice and Bob. 
Alice prepares a quantum system from $\en$ and the task of Bob is to guess the correct state. 
For a measurement $\M$ having the outcome set $X$, the guessing probability $\Pg(\mathcal{E};\M)$ is given as
\begin{equation*}
\Pg(\mathcal{E};\M) =  \sum_{x} \tr{\en(x) \M(x)} =  \sum_{x} p(x) \tr{\varrho_x \M(x)} \, .
\end{equation*}
The aim is to maximize the guessing probability, and we denote
\begin{equation}\label{eq:def_P(E)}
\Pg(\mathcal{E}) :=  \max_{\M}  \Pg(\mathcal{E};\M) \, ,
\end{equation}
where the optimization is over all measurements with outcome set $X$. 
This is called the optimal guessing probability for $\en$, and the minimum error discrimination problem is to find an optimizing measurement for a given state ensemble $\en$.
The problem was introduced in \cite{Holevo73,YuKeLa75,QDET76}. 
The existence of optimal measurements, i.e., the fact that in \eqref{eq:def_P(E)} the maximum is actually attained, follows by a compactness argument \cite[Proposition 4.1]{Holevo73}, \cite[Lemma 1]{YuKeLa75}.

\begin{figure}
\includegraphics[width=8cm]{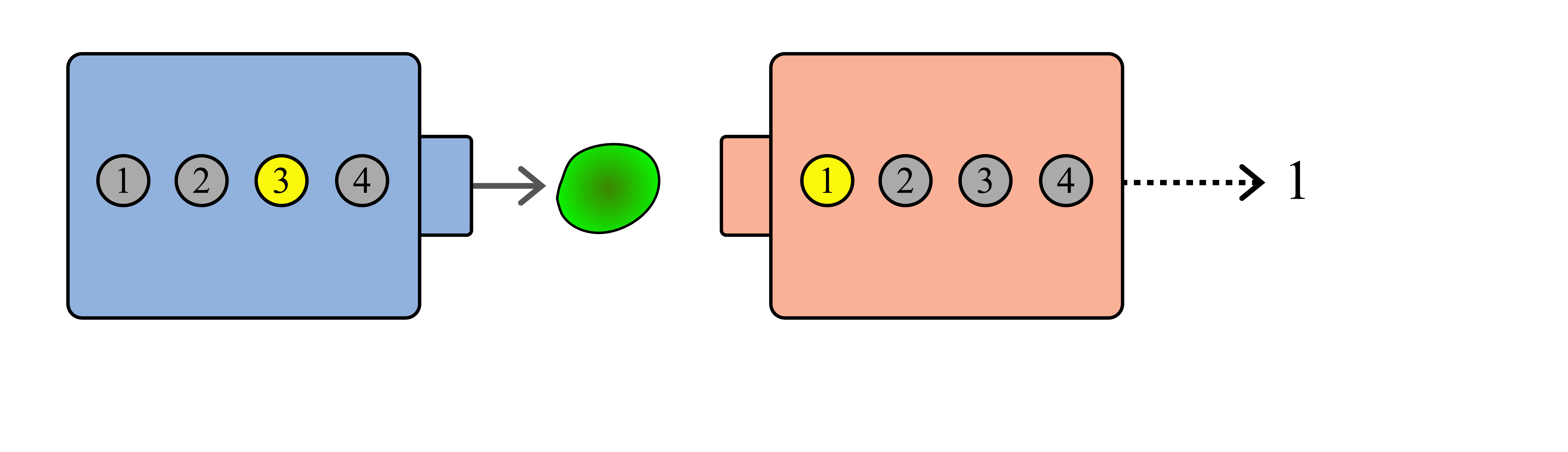}%(Color online) 
\caption{\label{fig:usual} In minimum error state discrimination, Alice prepares a quantum state \(\varrho_x\) from a given state ensemble which is known also to Bob. Bob aims to determine the label \(x\) by performing a measurement which maximizes his guessing probability. In the depicted event, Bob is making an incorrect guess.}
\end{figure}

In the state discrimination with post-measurement information, the standard scenario is modified by adding a middle step to it. 
The starting point, known both to Alice and Bob, is a  state ensemble $\en$ and a partition $\Pscr = (X_\ell)_{\ell\in I}$ of the label set $X$ into nonempty disjoint subsets. 
For each index $\ell\in I$, the probability that a label occurs in $X_\ell$ is
\begin{equation}\label{eq:q}
q(\ell) \coloneqq \sum_{x \in X_\ell} p(x) \,.
\end{equation}
We further assume that $q(\ell)\neq 0$ to avoid trivial cases. 
Then, conditioning the state ensemble $\en$ to the occurrence of a label in $X_\ell$, we obtain a new state ensemble $\en_\ell$, which we call a \emph{subensemble} of $\en$. The label set of $\en_\ell$ is $X_\ell$, and
\begin{equation}\label{eq:E_l}
\en_\ell(x) := \frac{1}{q(\ell)} \en(x) \, , \qquad x\in X_\ell \, .
\end{equation}

The steps in the scenario are the following (see Fig.~\ref{fig:post}(a)) :
\begin{enumerate}[(i)]
\item Alice picks a label $x$ from the set $X$, according to the prior probability distribution $p$.   
She then prepares a state $\varrho_{x}$ and delivers this state to Bob. 
\item Bob performs a measurement $\M$, hence obtaining an outcome $y\in Y$ with  probability $\tr{\varrho_{x} \M(y)}$. The outcome set $Y$ of $\M$ is freely chosen by Bob. 
\item After the measurement is performed, Alice tells to Bob the index $\ell$ of the correct subset $X_{\ell}$ where the  label was picked from.
\item Based on the measurement outcome $y$ and on the announced index $\ell$, Bob must guess $x$. 
This means that Bob applies a function $f_\ell:Y \to X_\ell$ to the obtained measurement outcome $y$ and his guess is $f_\ell(y)$.
\end{enumerate}

\begin{figure}
\centering
\subfigure[]
{
\includegraphics[width=9cm]{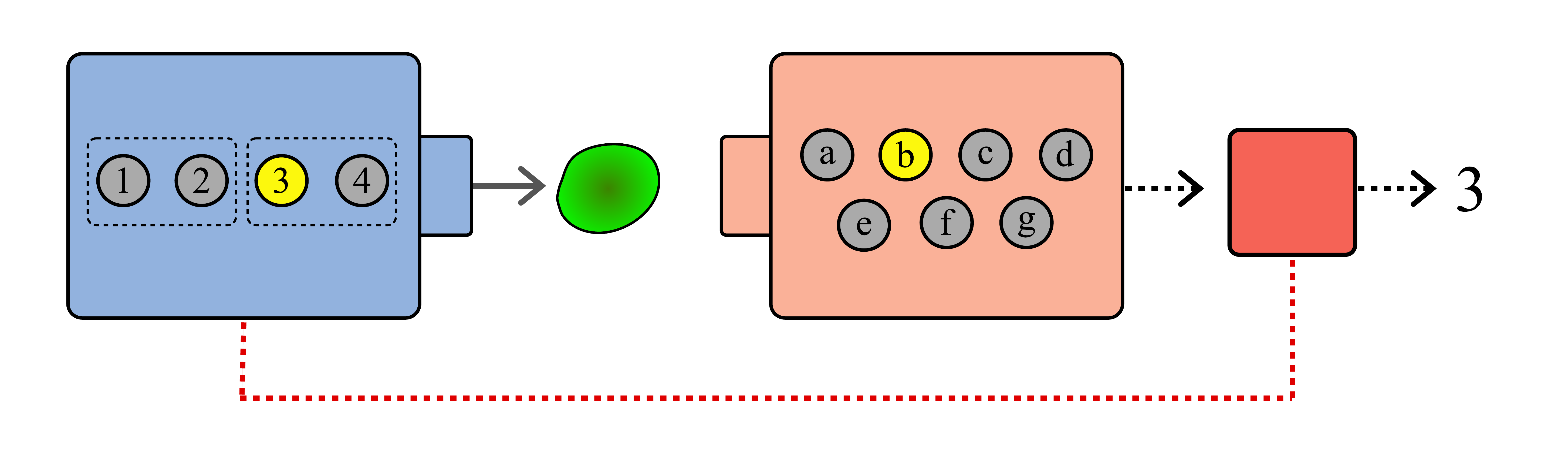}
}
\subfigure[]
{
\includegraphics[width=9cm]{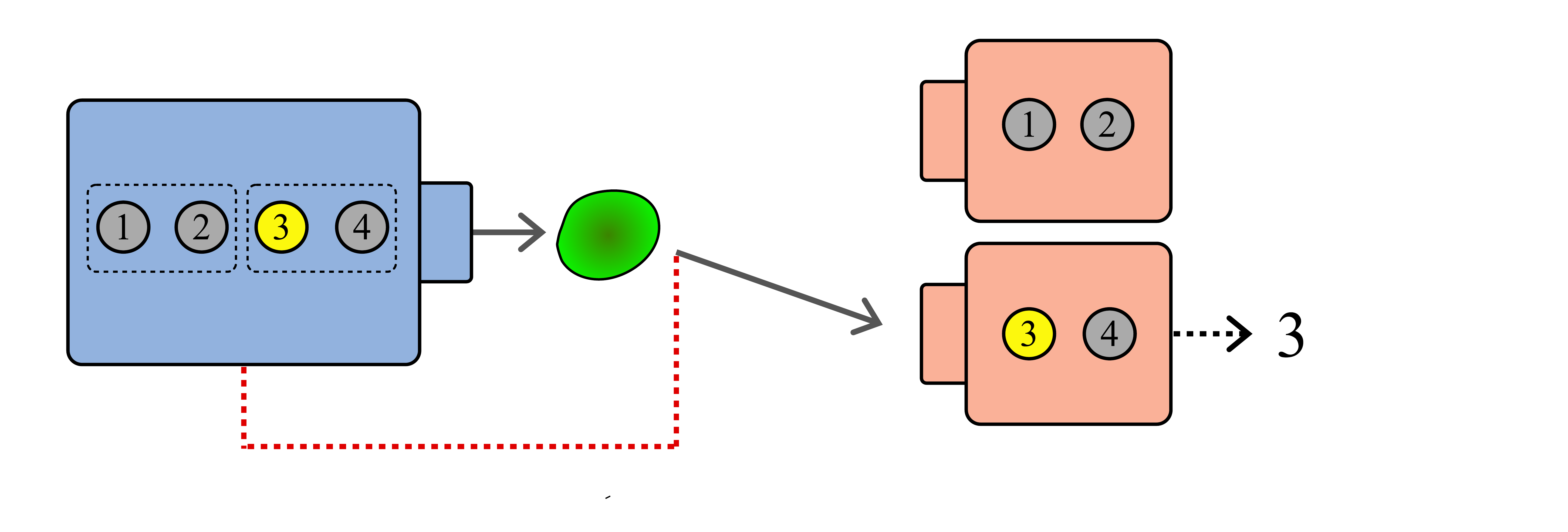}
}
\caption{\label{fig:post}%(Color online) 
Alice encodes a label $x$ into a quantum state \(\varrho_x\) on which Bob performs a measurement. In the {\em post-measurement} information scenario (a), Alice announces the subset from which she picked $x$ after Bob has performed a measurement. 
After Alice's announcement, Bob can post-processes his measurement outcome accordingly, and finally he gives his guess. In the {\em pre-measurement} information scenario (b), Alice announcez the correct subset already before Bob arranges his measurement.}
\end{figure}

Bob's guessing strategy is therefore determined by a measurement $\M$ and post-processing functions $(f_\ell)_{\ell\in I}$.
We emphasize that the same measurement $\M$ is used at every round, while the choice of the implemented relabeling function is determined by the announced label $\ell$.

We denote by $\Ppg(\en;\Pscr;\M,(f_\ell)_{\ell\in I})$ the guessing probability in the previously described scenario, and further,  we denote by $\Ppg(\en;\Pscr)$ the maximum  of the guessing probability when $\M$ and $(f_\ell)_{\ell\in I}$ vary over all suitable measurements and relabeling functions, respectively. 
Remarkably, optimal measurements and relabeling functions for the discrimination problem with post-measurement information actually exist, as we will see in Sec.~\ref{sec:reduction} below.

To elaborate the expression of $\Ppg(\en;\Pscr;\M,(f_\ell)_{\ell\in I})$, we denote by $f_{\ell\ast} \M$ the post-processed measurement that Bob has effectively performed when he has applied $f_\ell$ after $\M$, i.e., the measurement that has outcomes $X_\ell$ and is defined as  
\begin{align}
\label{eq:defpush}
f_{\ell\ast} \M ( x ) \coloneqq \sum_{y \in f_\ell^{-1}(x)} \M(y) \, ,  \qquad x\in X_\ell \, ,
\end{align}
where $f_\ell^{-1}(x)$ denotes the preimage of $x$, i.e., $f_\ell^{-1}(x)=\{ y: f_\ell(y)=x\}$.
We can then write the guessing probability as
\begin{equation}\label{eq:post}
\Ppg(\en;\Pscr;\M,(f_\ell)_{\ell\in I})  = \sum_{\ell\in I} q(\ell)\Pg(\en_{\ell};f_{\ell\ast} \M )  \, .
\end{equation}

The use of post-measurement information cannot decrease the guessing probability, that is,
\begin{equation}\label{eq:post-helps}
\Pg(\en)  \leq  \Ppg(\en;\Pscr)   \, .
\end{equation}
Indeed, one possible strategy for Bob is to perform a measurement $\M$ with outcomes in $X$ that optimally discriminates $\en$. He thus obtains the correct outcome with the probability $\Pg(\en)$, but he doesn't announce his guess yet. Then, after hearing the
index $\ell$ of the correct subset $X_\ell$, Bob does the following. If his obtained measurement outcome $x$ belongs
to $X_\ell$, then Bob's guess is $x$. But if $x$ is not in $X_\ell$, then Bob infers that he got an incorrect result and chooses an arbitrary default label $x_\ell\in X_\ell$ as his guess. This means that the restrictions $\left.f_\ell\right|_{X_\ell}$ of Bob's relabeling functions are the identity maps on $X_\ell$, and $f_\ell(x) = x_\ell$ whenever $x\notin X_\ell$. In this way, the post-measurement information allows Bob to sometimes neglect incorrect results, hence his guessing probability cannot be lower than $\Pg(\en)$. 
Formally, $f_{\ell\ast} \M ( x ) \geq \M ( x )$ for all $x\in X_\ell$, implying 
\begin{equation}
\Pg(\en_\ell ; f_{\ell\ast} \M) \geq \frac{1}{q(\ell)}\sum_{x\in X_\ell} \tr{\en(x)\M(x)} \, . 
\end{equation}
Using this inequality in \eqref{eq:post}, we get \eqref{eq:post-helps}.

From \eqref{eq:post} we also conclude a simple upper bound,
\begin{equation}\label{eq:before}
\Ppg(\en;\Pscr)  \leq \sum_{\ell\in I} q(\ell) \Pg(\en_{\ell} ) \, .
\end{equation}
The right hand side of \eqref{eq:before} is the optimal success probability if Alice would tell the used state ensemble to Bob \emph{before} Bob performs a measurement, in which case Bob can choose the optimal measurement to discriminate the correct state ensemble (see Fig.~\ref{fig:post}(b)).
We will thereby denote
\begin{equation}\label{eq:prior-def}
\Prg(\en;\Pscr) :=  \sum_{\ell\in I} q(\ell) \Pg(\en_{\ell} ) \, .
\end{equation}

In summary, the optimal guessing probability with post-measurement information is bounded in the interval
\begin{equation}
\Pg(\en) \leq \Ppg(\en;\Pscr) \leq \Prg(\en;\Pscr) \,,
\end{equation}
whose left and right extremes correspond to situations where Alice gives no information at all and Alice gives the partial information before Bob's choice of measurement, respectively.

%%%%%%%%%%%%%%%%%%%%%%
\subsection{Limiting to the standard form measurements}\label{sec:marginal}
%%%%%%%%%%%%%%%%%%%%%%

To maximize the guessing probability in the previously described scenario, Bob must find the optimal measurement $\M$ and relabeling functions $(f_\ell)_{\ell\in I}$.
The outcome set of $\M$ is, in principle, arbitrary and the role of relabeling functions is to adjust the obtained measurement outcome to give a meaningful guess.
However, as we will next show, there is a class of measurements with a fixed outcome set, determined by the separation of $X$ into subsets $(X_\ell)_{\ell\in I}$, such that we can always restrict the optimization to this class. 

A natural choice for the outcome set of Bob's mesurement is the Cartesian product $\bigtimes_{\ell\in I} X_\ell$, where $X_\ell$ is the label set of $\en_\ell$. 
For simplicity, in the following we assume the index set $I=I_m := \{1,\ldots,m\}$.
Then, at each measurement round Bob obtains a measurement outcome $(x_1,\ldots,x_m)$, and when Alice tells him the correct index $\ell$,  Bob just picks the outcome $x_{\ell}$  accordingly.
The respective relabeling function $f_\ell$ is now just the projection $\pi_\ell$ from $X_1\times\cdots\times X_m$ into $X_\ell$.
\emph{When Bob's measurement has the Cartesian product $X_1\times\cdots\times X_m$ as its outcome set, we will use the shorthand notation} 
\begin{equation*}
\Ppg(\en;\Pscr;\C) := \Ppg(\en;\Pscr;\C,(\pi_\ell)_{\ell\in I_m}) \, .
\end{equation*}
We thus have
\begin{equation}\label{eq:post-C}
\begin{aligned}
& \Ppg(\en;\Pscr;\C) = \sum_{\ell=1}^m q(\ell) \Pg(\en_{\ell};\pi_{\ell\ast} \C ) \\
& \qquad \qquad = \sum_{\ell=1}^m q(\ell) \sum_{x_\ell\in X_\ell} \tr{\en_\ell (x_\ell ) \pi_{\ell\ast} \C (x_\ell)} \, ,
\end{aligned}
\end{equation}
where, according to \eqref{eq:defpush},
\begin{equation}\label{eq:marginPOVM}
\pi_{\ell\ast} \C (x) = \sum_{\substack{x_1\in X_1,\ldots,x_m\in X_m\\ \text{such that } x_\ell = x}} \C(x_1,\ldots,x_m) \, .
\end{equation}

The next result justifies the choice of the Cartesian product.

\begin{proposition}\label{prop:standard}
For any choice of measurement $\M$ and relabeling functions $f_1,\ldots,f_m$, there is a measurement $\C$ with  product outcome set $X_1 \times \cdots \times X_m$ such that 
\begin{align*}
\Ppg(\en;\Pscr;\M,(f_\ell)_{\ell\in I_m}) & = \Ppg(\en;\Pscr;\C) \, .
\end{align*}
\end{proposition}

\begin{proof}
We define $\C$ as
\begin{align*}
\C(x_1,\ldots,x_m) & = \sum_{y \in f_1^{-1}(x_1) \cap \cdots \cap f_m^{-1}(x_m)} \M(y) \,.
\end{align*}
Then, by \eqref{eq:marginPOVM},
\begin{equation*}
\pi_{\ell\ast} \C(x) = \sum_{y\in f_\ell^{-1}(x)} \M(y) = f_{\ell\ast} \M ( x ) \, ,
\end{equation*}
hence, $\pi_{\ell\ast} \C = f_{\ell\ast} \M$. 
\end{proof}

%%%%%%%%%%%%%%%%%%%%%%
\subsection{Remarks on other formulations of the problem}\label{sec:comparison}
%%%%%%%%%%%%%%%%%%%%%%

The problem of state discrimination with post-measurement information was first considered by Winter, Ballester and Wehner in \cite{BaWeWi08}. According to their approach, before Alice announces the subensemble the state was picked from, Bob is allowed to store both classical and quantum information; his classical resources are unlimited (an unbounded amount of classical memory), and on the quantum side he can use a string of qubits with prescribed length. Later, Gopal and Wehner (GW) restricted to the case where only classical information is available for Bob \cite{GoWe10}; for this reason, their approach more directly compares with ours.

In GW's problem, Alice encodes a string $x$ of classical information in one of $m$ possible quantum states $(\varrho_{x,b})_{b\in\ca{B}}$, with $\ca{B}=\{1,\ldots,m\}$.
The aim of Bob is to determine the string \(x\), irrespectively of the encoding chosen by Alice (see Fig.~\ref{fig:gw}). The set $\ca{X}$ from which $x$ is picked is the same for all encodings $b\in\ca{B}$, while the probability of selecting a specific encoding $b\in\ca{B}$ may depend on the chosen $x$. Thus, if $p(x,b)$ is the joint probability of picking the string $x$ and using the encoding $b$, Bob's received state is the mixture $\sum_{x,b} p(x,b)\varrho_{x,b}$. On this state, Bob performs a measurement $\C$ with outcomes in the Cartesian product $\ca{X}^m$, thus obtaining the result $(x_1,\ldots,x_m)$. Then, Alice declares him the selected encoding $b$, and, according to the announced $b$, Bob guesses the value $x_b$ for the string $x$. Clearly, also in this scenario, Bob's maximum success probability with post-measurement information $p^I_{\rm succ}$ can not be smaller than the analogous probability without post-measurement information $p_{\rm succ}$. When $p^I_{\rm succ} \equiv p_{\rm succ}$, \emph{post-measurement information is useless} for the encoding at hand; in this case, if a measurement $\M$ with outcomes in $\ca{X}$ is optimal for the problem without post-measurement information, then the diagonal measurement 
$$
\C(x_1,\ldots,x_m) = \begin{cases}
\M(x_0) & \text{ if $x_1 = x_2 = \ldots = x_m \equiv x_0$}\\
0 & \text{otherwise}
\end{cases}
$$
is optimal for the problem with post-measurement information. Diagonal measurements correspond to the situation in which Bob guesses the same string $x_0$ independently of Alice's announced encoding, i.e., he completely ignores post-measurement information.

\begin{figure}
\includegraphics[width=9cm]{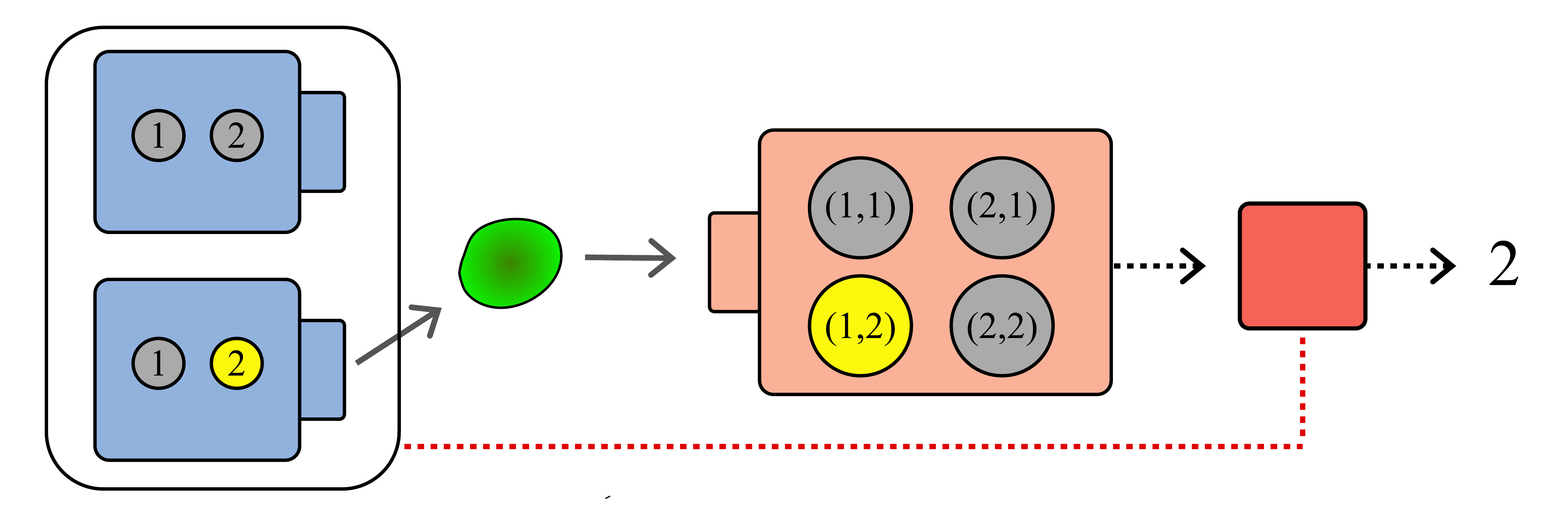}%(Color online) 
\caption{\label{fig:gw} In the approach described in \cite{GoWe10},  Alice encodes a classical string \(x\) into a quantum state \(\varrho_{x,b}\), where \(b\) specifies one of the possible encodings Alice can choose from. Bob must determine the string \(x\) irrespectively of the encoding choosen by Alice, but Alice announces the encoding after Bob has performed his measurement.}
\end{figure}

To cast GW's approach into our framework, we choose as our label set $X$ the disjoint union of $m$ copies of $\ca{X}$, i.e., $X = \{(x,b) : x\in\ca{X},\ b\in\ca{B}\}$, and we consider the state ensemble $\en(x,b) = p(x,b)\varrho_{x,b}$. For all $b\in\ca{B}$, we denote $X_b = \{(x,b) : x\in\ca{X}\}$, so that the sets $(X_b)_{b\in\ca{B}}$ constitute a partition $\Pscr$ of $X$. Then, Bob's task of identifying the string $x$ with post-measurement information $b$ in GW's scenario is equivalent to the corresponding problem of detecting the label $(x,b)$ within our approach; in particular, $p^I_{\rm succ} = \Ppg(\en;\Pscr)$. Note that GW actually do not consider Bob's possibility to arbitrarily enlarge his classical memory (i.e., his outcome set $Y$), as they directly set $Y = \ca{X}^m$; as we have proved in Prop.~\ref{prop:standard}, this assumption is not restrictive.

However, it is important to stress that the success probabilities without post-measurement information can differ in the two approaches; indeed, $p_{\rm succ} \geq \Pg(\en)$, with strict inequality in many concrete examples. 
This is due to the fact that in GW's setting Bob is required to guess only the string $x$, while with our definition of $\Pg(\en)$ we require Bob to guess the whole label $(x,b)$, i.e., both the string $x$ and the encoding $b$ selected by Alice. 
For this reason, there are situations in which post-measurement information is useless for GW's approach, although we have $\Ppg(\en;\Pscr)>\Pg(\en)$. 
For further discussion on this point, we defer to the examples in Sec.~\ref{sec:qubit}.

%%%%%%%%%%%%%%%%%%%%%%
\section{Post-measurement information and incompatibility of measurements}
%%%%%%%%%%%%%%%%%%%%%%

%%%%%%%%%%%%%%%%%%%%%%
\subsection{Compatible measurements}\label{sec:compatible}
%%%%%%%%%%%%%%%%%%%%%%

As we see from \eqref{eq:post-C}, the guessing probability $\Ppg(\en;\Pscr;\C)$  depends only on the relabeled measurements $\pi_{1\ast}\C,\ldots,\pi_{m\ast}\C$, not on other details of $\C$.
The measurement $\pi_{\ell\ast}\C$, given by \eqref{eq:marginPOVM}, is called the \emph{$\ell$th marginal of $\C$}.
This way of writing reveals immediately the connection with the compatibility of measurements.
Namely, we recall that measurements $\N_1,\ldots,\N_m$ are called \emph{compatible} (also \emph{jointly measurable}) if there exists a measurement $\M$ on their Cartesian product outcome set such that each measurement $\N_\ell$ is the $\ell$th marginal of $\M$. 
We remark that Prop.~\ref{prop:standard} can also be extracted from the fact that the functional coexistence relation is equivalent to the compatibility relation \cite{LaPu97}. 
We further note that, by applying the equivalent definition of compatibility in terms of the post-processing preorder \cite{AlCaHeTo09,HeMiZi16}, we conclude that allowing non-deterministic post-processing functions does not increase the optimal guessing probability $\Ppg(\en;\Pscr)$.

Combining \eqref{eq:post-C}, \eqref{eq:marginPOVM} and Prop.~\ref{prop:standard} allows us to write the optimal guessing probability with post-measurement information as follows:
\begin{equation}\label{eq:max-comp}
\begin{split}
& \Ppg(\en;\Pscr) = \\
& \ \ \max \Big\{  \sum_{\ell=1}^m q(\ell) \Pg(\en_{\ell}; \N_\ell ) : \N_1,\ldots,\N_m \textrm{ compatible} \Big\}  \, .
\end{split}
\end{equation}
We now see that the difference between the guessing probabilities in prior and posterior information scenarios is that in the first one the optimization over measurements $\N_1,\ldots,\N_m$ has no restrictions, while in the second one they must be compatible.
This leads to the following conclusion.

\begin{theorem}\label{prop:equivalence}
There exist compatible optimal measurements
for the discrimination problems of state ensembles $\en_{1},\ldots,\en_{m}$ if and only if the posterior and prior information discrimination problems have the same optimal guessing probability, i.e., 
\begin{align}\label{eq:equal}
\Ppg(\en;\Pscr) = \Prg(\en;\Pscr) \, .
\end{align}
\end{theorem}

\begin{proof}
It follows from the definition of \(\Prg(\en;\Pscr)\) and \eqref{eq:max-comp} that,
if there exists compatible optimal measurements $\N_1,\ldots,\N_m$, then \eqref{eq:equal} holds.
Let us then assume that \eqref{eq:equal} holds. 
This means that there exist compatible measurements $\N_1,\ldots,\N_m$ such that $\sum_{\ell=1}^m q(\ell) \Pg(\en_{\ell}; \N_\ell ) =  \sum_{\ell=1}^m q(\ell) \Pg(\en_{\ell} )$. 
Since for any $\ell$ we have $\Pg(\en_{\ell}; \N_\ell ) \leq  \Pg(\en_{\ell} )$, the previous equality and $q(\ell)\neq 0$ for all $\ell$ imply $\Pg(\en_{\ell}; \N_\ell ) =  \Pg(\en_{\ell} )$.
Therefore, each $\N_\ell$ is an optimal measurement for the discrimination problem of $\en_{\ell}$. 
\end{proof}

We recall that a minimum error discrimination problem may not have a unique optimal measurement. 
For the statement of Prop.~\ref{prop:equivalence} it is enough that at least one collection of optimal measurements is made up of compatible measurements.

%%%%%%%%%%%%%%%%%%%%%%
\subsection{Incompatible measurements}\label{sec:incompatible}
%%%%%%%%%%%%%%%%%%%%%%

We now turn into the case when optimal measurements $\N_1,\ldots,\N_m$ for the standard minimum error discrimination for the state ensembles $\en_{1},\ldots,\en_{m}$ are incompatible.
From Theorem \ref{prop:equivalence} we conclude that in this case $\Ppg(\en;\Pscr)$ is strictly smaller than $\Prg(\en;\Pscr)$. 
However, we can still ask if the optimal solutions for the discrimination of subensembles $\en_\ell$ give some hint on the optimal solution for the post-measurement information discrimination.

A heuristic approach to the problem of state discrimination with post-measurement information relies on \eqref{eq:max-comp} and goes as follows. 
We form a noisy version $\nN_\ell$ of each optimal measurement $\N_\ell$ related to $\en_\ell$, and we add enough noise to make the measurements $\nN_1,\ldots,\nN_m$ compatible. 
Noisy versions can be, in principle, any measurements that are compatible but are approximating the optimal measurements $\N_1,\ldots,\N_m$ reasonably well. 
Bob then performs a joint measurement of $\nN_1,\ldots,\nN_m$ and from here on out, he follows the same procedure as in the case of compatible measurements.
One would expect the guessing probability to be relatively good if $\nN_1,\ldots,\nN_m$ are good approximations of $\N_1,\ldots,\N_m$.

One type of noisy version of a measurement $\N_\ell$ is given by the mixture
\begin{equation}\label{eq:noisy}
\nN_\ell(x) = t_\ell \N_\ell(x) + (1-t_\ell) \nu_\ell(x) \id \, , 
\end{equation}
where $\nu_\ell$ is a probability distribution and $t_\ell\in [0,1]$ is a mixing parameter. 
We then have
\begin{equation}\label{eq:noisy-guessing}
\Pg(\en_\ell;\nN_\ell) \geq  t_\ell \Pg(\en_\ell)  \, .
\end{equation}
One would aim to choose each mixing parameter $t_\ell$ as close to $1$ as possible to make $\nN_\ell$ a good approximation of $\N_\ell$, but the requirement that $\nN_1,\ldots,\nN_m$ must be compatible limits the region of the allowed tuples $(t_1,\ldots,t_m)$. 
The set of all tuples $(t_1,\ldots,t_m)$ that make the mixtures \eqref{eq:noisy} compatible for some choices of $\nu_1,\ldots,\nu_m$ is called the \emph{joint measurability region of $\N_1,\ldots,\N_m$} \cite{BuHeScSt13}, and we denote it as $J(\N_1,\ldots,\N_m)$. 
Further, the greatest number $t$ such that $(t,\ldots,t) \in J(\N_1,\ldots,\N_m)$ is called the \emph{joint measurability degree of $\N_1,\ldots,\N_m$} \cite{HeScToZi14}, and we denote it as $\jmd(\N_1,\ldots,\N_m)$. 

The choice of the most favorable tuple $(t_1,\ldots,t_m)\in J(\N_1,\ldots,\N_m)$ for the discrimination with post-measurement information depends on the probability distribution $q$ and on the optimal guessing probabilities $\Pg(\en_\ell)$. 
Starting from \eqref{eq:max-comp} and using \eqref{eq:noisy-guessing}, we obtain a lower bound
\begin{align*}
& \Ppg(\en;\Pscr) \\
& \geq \max \Big\{  \sum_{\ell=1}^m t_\ell q(\ell) \Pg(\en_{\ell}) : \\
& \qquad \qquad \qquad \qquad \quad (t_1,\ldots,t_m) \in J(\N_1,\ldots,\N_m) \Big\} \\
& \geq \jmd(\N_1,\ldots,\N_m) \cdot \sum_{\ell=1}^m q(\ell) \Pg(\en_{\ell}) \\
& = \jmd(\N_1,\ldots,\N_m) \cdot \Prg(\en_{\ell};\Pscr)  \, .
\end{align*}
The joint measurability degree of a set of observables is one if and only if the observables are compatible, therefore the obtained inequality can be taken as quantitative addition to Theorem \ref{prop:equivalence}. 

To derive another related inequality, we consider noisy versions of the form
\begin{equation}\label{eq:noisy-uniform}
\nN_\ell(x) = t_\ell \N_\ell(x) + (1-t_\ell) \,\frac{1}{n_\ell} \id \, , 
\end{equation}
where $n_\ell$ is the number of elements of $X_\ell$. 
Compared to the more general form \eqref{eq:noisy}, the added noise is here given by a uniform probability distribution. 
We denote by $J_u(\N_1,\ldots,\N_m)$ and $\jmdu(\N_1,\ldots,\N_m)$ the analogous objects as $J(\N_1,\ldots,\N_m)$ and $\jmd(\N_1,\ldots,\N_m)$, but where the added noise is given by uniform probability distributions.   
Clearly, $\jmdu(\N_1,\ldots,\N_m) \leq \jmd(\N_1,\ldots,\N_m)$, but the benefit for the current task is that now we can calculate the exact relation between $\Pg(\en_\ell;\nN_\ell)$ and $\Pg(\en_\ell)$. 
Namely, the bound \eqref{eq:noisy-guessing} is replaced by
\begin{equation}\label{eq:noisy-guessing-u}
\Pg(\en_\ell;\nN_\ell) =  t_\ell \Pg(\en_\ell) + (1-t_\ell) \,\frac{1}{n_\ell}  \, ,
\end{equation}
and the additional second term may improve the earlier bounds. 
For example, in the special case when $\Pg(\en_\ell)=1$ and $n_\ell \equiv n$ for each $\ell=1,\ldots,m$, we get
\begin{align}\label{eq:jmd-bound}
\Ppg(\en;\Pscr) \geq \frac{1}{n} + \frac{n-1}{n} \cdot \jmdu(\N_1,\ldots,\N_m)  \, . 
\end{align}
Similar lower bounds can be calculated in other cases. 

%%%%%%%%%%%%%%%%%%%%%%
\subsection{Approximate cloning strategy}\label{sec:cloning}
%%%%%%%%%%%%%%%%%%%%%%

Approximate cloning device is, generally speaking, a physically realizable map that makes several approximate copies from an unknown quantum state. 
One such device is Keyl-Werner cloning device \cite{Werner98,KeWe99}, which takes an unknown state $\varrho$ as input and outputs $m$ approximate copies $\tilde{\varrho}$ of the form 
\begin{equation*}
\tilde{\varrho} = c_{m,d}\,\varrho + (1-c_{m,d} )\, \frac{1}{d} \,\id \, , \qquad c_{m,d} = \frac{m+d}{m(1+d)} \, .
\end{equation*} 
This device is known to be optimal if the quality of single clones is quantified as their fidelity with respect to the original state.

In the current scenario of state discrimination with post-measurement information, we can use an approximate cloning device in the following way (see Fig.~\ref{fig:cloning}).
Bob, after receiving a quantum system from Alice, copies the unknown state approximatively into $m$ systems. 
For each copy, Bob performs the measurement $\N_\ell$ that optimally discriminates the subensemble $\en_\ell$. 
Then, after Alice announces the index $\ell$ of the correct subset $X_{\ell}$ where the label was from, Bob chooses his guess accordingly.

This cloning strategy is rarely optimal (see examples in Sec.~\ref{sec:qubit}), but it gives a non-trivial lower bound for the guessing probability $\Ppg(\en;\Pscr)$.
For instance, if the prior probability distribution $p$ is uniform, the cloning strategy leads to the lower bound
\begin{align}\label{eq:cloning}
\Ppg(\en;\Pscr)
 \geq c_{m,d} \, \Prg(\en;\Pscr) + 
(1-c_{m,d}) \frac{m}{N} \, ,
\end{align}
where $m$ is the number of blocks in the partition and $N=n_1 + \cdots + n_m$ is the total size of the index set. 

\begin{figure}
\includegraphics[width=9cm]{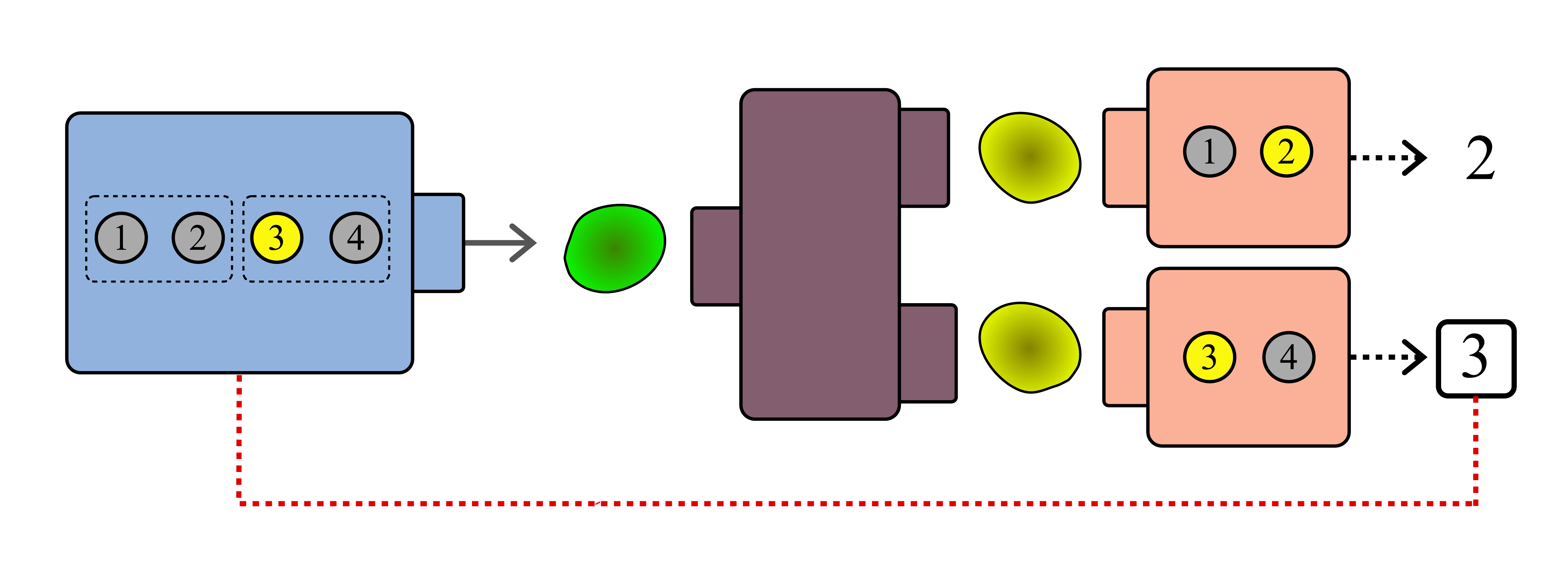}%(Color online) 
\caption{\label{fig:cloning}Approximate cloning can be used to obtain lower bounds for the post-information guessing probability.  Alice sends the state \(\varrho_x\) from the subsensemble \(\en_\ell\) to Bob. He then makes an approximate cloning of \(\varrho_x\) into as many copies as the number of subensembles, and, on each approximate copy, he performs the  measurement  that optimally discriminates  the corresponding  subensemble. Finally, after Alice announces the index \(\ell\), Bob chooses his guess accordingly.}
\end{figure}

We can also think the approximate cloning in the Heisenberg picture, and looking in that way the Keyl-Werner cloning device transforms each measurement $\N_\ell$ into
\begin{equation}\label{eq:cloning-obs}
\nN_\ell(x) = c_{m,d} \N_\ell (x) + (1-c_{m,d}) \tr{\N_\ell(x)} \, \frac{1}{d} \,\id \, .
\end{equation}
From this point of view, the approximative cloning strategy is just a particular instance of the noisy joint measurement strategy described in Sec.~\ref{sec:incompatible}; the lower bound \eqref{eq:cloning} then follows just by inserting \eqref{eq:cloning-obs} into the right hand side of \eqref{eq:max-comp}. 
The bound \eqref{eq:cloning} is useful as it is universal, in the sense that it does not depend on any details of the optimal measurements $\N_1,\ldots,\N_m$.

%%%%%%%%%%%%%%%%%%%%%%
\section{Methods to calculate the optimal guessing probability}
%%%%%%%%%%%%%%%%%%%%%%

%%%%%%%%%%%%%%%%%%%%%%
\subsection{Reduction to usual state discrimination problem}\label{sec:reduction}
%%%%%%%%%%%%%%%%%%%%%%

It was noted in \cite{GoWe10} that the state discrimination with post-measurement information problem can be related to a suitable standard state discrimination problem. 
Here we provide a slightly different viewpoint on this connection. 

As before, we consider a state ensemble $\en$ with label set $X$, and a partition $\Pscr = (X_\ell)_{\ell\in I_m}$ of $X$ into $m$ nonempty disjoint subsets.
As shown in Sec.~\ref{sec:marginal}, in order to maximize the posterior information guessing probability $\Ppg(\en;\Pscr;\M,(f_\ell)_{\ell\in I_m})$ over all measurements $\M$ and relabeling functions $f_1,\ldots,f_m$, it is enough to consider all measurements $\C$ with the Cartesian product outcome space and use the fixed post-processings $\pi_1,\ldots,\pi_m$. 
It turns out that, up to a constant factor, the guessing probability $\Ppg(\en;\Pscr;\C)$ is the same as the guessing probability for a certain specific state ensemble in the standard state discrimination scenario using the same measurement $\C$.
To explain the details of this claim, we define an auxiliary state ensemble $\enf$ having the Cartesian product $X_1\times\cdots\times X_m$ as its label set, and given by
\begin{equation}\label{eq:assisting}
\enf(x_1,\ldots,x_m) =\frac{1}{\Delta} \sum_{\ell=1}^m \en(x_\ell) = \frac{1}{\Delta} \sum_{\ell=1}^m q(\ell) \en_\ell(x_\ell) \, ,
\end{equation}
where the probability $q$ and the state ensembles $\en_\ell$ are defined in \eqref{eq:q}, \eqref{eq:E_l}, and the numerical factor $\Delta$ is
\begin{align}
\Delta \equiv \Delta(q;n_1,\ldots,n_m) = n_1\cdots n_m \sum_{\ell=1}^m \frac{q(\ell)}{n_\ell} \, .
\end{align}
(We recall that $n_\ell$ denotes the number of labels in $X_\ell$.)
The state ensemble $\enf$ has $n_1\cdots n_m$ labels and its states are convex combinations of states from different subensembles $\en_{\ell}$.  
Starting from \eqref{eq:post-C}, a direct calculation gives
\begin{align*}
& \Ppg(\en;\Pscr;\C) = \sum_{\ell=1}^m q(\ell) \sum_{x_\ell \in X_\ell} \tr{\en_\ell(x_\ell ) \pi_{\ell\ast} \C (x_\ell)}\\
& \quad = \sum_{\ell=1}^m q(\ell) \sum_{x_1\in X_1,\ldots,x_m\in X_m} \tr{\en_\ell(x_\ell ) \C(x_1,\ldots,x_m) }\\
& \quad = \sum_{x_1\in X_1,\ldots,x_m\in X_m} \tr{\sum_{\ell=1}^m q(\ell) \en_\ell(x_\ell ) \C(x_1,\ldots,x_m) }\\
& \quad = \Delta \cdot \Pg(\enf;\C) \, .
\end{align*}
The factor $\Delta$ is required as the state ensemble $\enf$ must be normalized, i.e.,
\begin{equation}
\sum_{x_1\in X_1,\ldots,x_m\in X_m} \tr{\enf(x_1,\ldots,x_m)} = 1 \, .
\end{equation}

As mentioned in the Sec.~\ref{sec:scenario},  it is known that the standard discrimination guessing probability \(\Pg(\enf;\C)\) always attains the maximum. From the previous connection we can conclude that the same holds for the post-measurement information problem.

The above discussion is summarized in the following result.
\begin{theorem}\label{prop:equiv}
The posterior information guessing probability $\Ppg(\en;\Pscr;\C)$ attains its maximum value when $\C$ is the optimal measurement for the standard discrimination problem of the state ensemble $\enf$ in \eqref{eq:assisting}.
The optimal guessing probabilities are related via the equation
\begin{equation*}
\Ppg(\en;\Pscr) = \Delta(q;n_1,\ldots,n_m) \cdot \Pg(\enf) \,.
\end{equation*}
\end{theorem}

As an illustration, suppose $X$ has $2n$ elements and it is partitioned into $X_1$ and $X_2$, both having $n$ elements, and that the prior probability $p$ is the uniform distribution on $X$.
Then
\begin{align*} 
\enf(x_1,x_2) = \frac{1}{n^2} \cdot \frac{1}{2} (\varrho_{x_1} + \varrho_{x_2})\, .
\end{align*}
We thus see that the state ensemble $\enf$ contains all possible equal mixtures of states from $\en_{1}$ and $\en_2$.

%%%%%%%%%%%%%%%%%%%%%%
\subsection{Optimal guessing probability in the usual state discrimination problem}\label{sec:methods}
%%%%%%%%%%%%%%%%%%%%%%

We have just seen that it is always possible to transform the problem of state discrimination with post-measurement information to a usual minimum error state discrimination problem. For this reason, in this section we consider a class of cases where for a single state ensemble $\en$ one can analytically calculate the optimal guessing probability $\Pg(\en)$ as well as the optimal measurements. 
This covers the cases that we will present as examples in Secs.~\ref{sec:qubit} and \ref{sec:mub}.

The main result is the following observation.

\begin{proposition}\label{prop:easy}
Suppose $\en$ is a state ensemble with label set $X$. For all $x\in X$, denote by $\lambda(x)$ the largest eigenvalue of $\en(x)$, and by $\Pi(x)$ the orthogonal projection onto the $\lambda(x)$-eigenspace of $\en(x)$. Define
\begin{equation*}
\lambda_\en = \max_{x\in X} \lambda(x) \, , \qquad X_\en = \{ x \in X: \lambda(x) = \lambda_\en \} \, .
\end{equation*}
Then, if there exists $\mu\in\real$ such that
\begin{equation}\label{eq:easy_condition_1}
\sum_{x\in X_\en} \Pi(x) = \mu\id \,,
\end{equation}
we have the following consequences:
\begin{enumerate}[(a)]
\item $\mu = \tfrac{1}{d} \sum_{x\in X_\en} \rank{\Pi(x)}$;\label{item:(mu)prop_easy}
\item $\Pg(\en) = d\lambda_\en$;\label{item:(lambda)prop_easy}
\item a measurement $\M_0$ attaining the maximum guessing probability $\Pg(\en)$ is\label{item:(M0)prop_easy}
\begin{equation}\label{eq:M0}
\M_0(x) = \begin{cases}
\mu^{-1}\Pi(x) & \text{ if } x\in X_\en \\
0 & \text{ if } x\notin X_\en \,;
\end{cases}
\end{equation}
\item a measurement $\M$ attains the maximum guessing probability $\Pg(\en)$ if and only if\label{item:(max)prop_easy}
\begin{enumerate}[(i)]
\item $\M(x) \leq \Pi(x)$ for all $x\in X_\en$
\item $\M(x) = 0$ for all $x\notin X_\en$.
\end{enumerate}
\end{enumerate}
\end{proposition}

In the following we provide a simple proof of Prop.~\ref{prop:easy}, relying on Lemma \ref{lem:easy} given after that. 
We remark that an alternative longer proof can also be given by making use of the optimality conditions \cite[Eq.~(III.29)]{YuKeLa75}, \cite[Theorem II.2.2]{Holevo78}, which follow from a semidefinite programming argument (see also \cite{ElMeVe03} for a more recent account of these results).

\begin{proof}
We assume that \eqref{eq:easy_condition_1} holds for some $\mu\in\real$. 
By taking the trace of both sides of \eqref{eq:easy_condition_1}, we get $\mu = \tfrac{1}{d} \sum_{x\in X_\en} \rank{\Pi(x)}$. This proves \eqref{item:(mu)prop_easy}.

For any measurement $\M$ on $X$, we have
\begin{align*}
\Pg(\en;\M) & = \sum_{x \in X} \tr{\en(x) \M(x)} \leq \sum_{x \in X} \lambda(x) \tr{\M(x)} \\
& \leq \lambda_\en \sum_{x \in X} \tr{\M(x)} = \lambda_\en \tr{\id} = d\lambda_\en \,. 
\end{align*}
The first inequality follows from Lemma \ref{lem:easy} just below, which also implies that the equality is attained if and only if $\M(x)\leq\Pi(x)$ for all $x\in X$. The second inequality is trivial, and it becomes an equality if and only if $\M(x) = 0$ for all $x\notin X_\en$. In summary, $\Pg(\en;\M) \leq d\lambda_\en$, with equality if and only if the measurement $\M$ satisfies conditions (i) and (ii) of \eqref{item:(max)prop_easy}.

Since $\Pi(x)\leq \mu\id$ for any $x\in X_\en$ by \eqref{eq:easy_condition_1}, we must have $\mu\geq 1$.
Hence, $\M_0(x)\leq\Pi(x)$ for all $x\in X_\en$.
Moreover, $\M_0(x) = 0$ for all $x\notin X_\en$ by the definition of $\M_0$. By the discussion in the last paragraph, it follows that $\M_0$ is optimal, and $\Pg(\en) = \Pg(\en ; \M_0) = d\lambda_\en$. This proves \eqref{item:(lambda)prop_easy} and \eqref{item:(M0)prop_easy}. Since any optimal measurement $\M$ must then be such that $\Pg(\en ; \M) = d\lambda_\en$, also \eqref{item:(max)prop_easy} follows.
\end{proof}

\begin{lemma}[for Prop.~\ref{prop:easy}]\label{lem:easy}
Let $A,B\in\lh$ with $A\geq 0$ and $0\leq B\leq\id$. 
Let $\lambda$ be the largest eigenvalue of $A$ and $\Pi$ the associated eigenprojection. Then,
\begin{equation*}
\tr{AB} \leq \lambda \, \tr{B} \,,
\end{equation*}
and the equality is attained if and only if $B\leq\Pi$.
\end{lemma}

\begin{proof}
Since $\lambda\id-A\geq 0$, we have $\lambda\tr{B} - \tr{AB} = \tr{(\lambda\id-A)B} \geq 0$, where the inequality follows from \cite[Lemma 2]{YuKeLa75}. By the same result, the equality is attained if and only if $(\lambda\id-A)B =0$, that is, $AB = \lambda B$. Note that $AB = \lambda B \ \Leftrightarrow \ \ran{B}\subseteq\ran{\Pi}$. The latter inclusion implies $\Pi B = B = B\Pi$ and then $B = \Pi B\Pi \leq \Pi\id\Pi = \Pi$. Conversely, if $B\leq\Pi$, then $\ker{\Pi}\subseteq\ker{B}$, so that $\ran{B}\subseteq\ran{\Pi}$. In conclusion, $AB=\lambda B$ if and only if $B\leq\Pi$, and this completes the proof.
\end{proof}

\begin{corollary}\label{prop:easy-2}
With the notations of Prop.~\ref{prop:easy}, suppose \eqref{eq:easy_condition_1} holds for some $\mu\in\real$ and $\rank{\Pi(x)}=1$ for all $x\in X_\en$. Then, the following facts are equivalent:
\begin{enumerate}[(i)]
\item The operators $\{\Pi(x) : x\in X_\en\}$ are linearly independent.\label{it:(i)_prop:easy-2}
\item The measurement $\M_0$ given in \eqref{eq:M0} is the unique measurement giving the maximum guessing probability $\Pg(\en)$.\label{it:(ii)_prop:easy-2}
\end{enumerate}
\end{corollary}
\begin{proof}
Since $\Pi(x)$ is a rank-1 orthogonal projection, any positive operator $A$ satisfying $A \leq \Pi(x)$ is a scalar multiple of $\Pi(x)$. 
Therefore, by \eqref{item:(max)prop_easy} of Prop.~\ref{prop:easy}, a measurement attains the maximum guessing probability $\Pg(\en)$ if and only if it has the form
\begin{equation*}
\M_\alpha(x) = \begin{cases}
\alpha(x)\Pi(x) & \text{ if } x\in X_\en \\
0 & \text{ if } x \notin X_\en
\end{cases}
\end{equation*}
for some function $\alpha : X_\en \to [0,1]$.\\
Since 
\begin{equation*}
\id = \sum_{x \in X} \M_\alpha(x) = \sum_{x\in X_\en} \alpha(x) \Pi(x)
\end{equation*}
and 
\begin{equation*} 
 \id = \sum_{x \in X} \M_0(x) = \sum_{x\in X_\en} \mu^{-1} \Pi(x) \,,
\end{equation*} 
linear independence of the operators $\{\Pi(x) : x\in X_\en\}$ yields $\alpha(x) = \mu^{-1}$ for all $x\in X_\en$, hence $\M_\alpha=\M_0$.\\
Conversely, if the operators $\{\Pi(x) : x\in X_\en\}$ are not linearly independent, then $\mu>1$, as otherwise they would be an orthogonal resolution of the identity, that is a contradiction. Moreover, there exists some nonzero function $\beta : X_\en\to\complex$ such that
$$
0 = \sum_{x\in X_\en} \beta(x) \Pi(x) = \sum_{x\in X_\en} \overline{\beta(x)} \Pi(x) \,.
$$
By possibly replacing $\beta$ with either $\beta+\overline{\beta}$ or $\i (\beta-\overline{\beta})$, we can assume that $\beta : X_\en\to\real$. If $\epsilon\in\real\setminus\{0\}$ is such that $\abs{\epsilon}$ is small enough, then $\alpha(x) = \mu^{-1}+\epsilon\beta(x)\in [0,1]$ for all $x\in X_\en$; hence $\M_\alpha$ is an optimal measurement with $\M_\alpha\neq\M_0$.
\end{proof}

We remark that if the rank-$1$ condition in the statement of Cor.~\ref{prop:easy-2} is dropped, then the equivalence of items \eqref{it:(i)_prop:easy-2} and \eqref{it:(ii)_prop:easy-2} is no longer true; a simple example demonstrating this fact is provided in Appendix \ref{app:one}.

\begin{corollary}\label{prop:easy-3}
Suppose $\Pscr = (X_\ell)_{\ell\in I}$ is a partition of $X$ into nonempty disjoint subsets, and define $q$ and $\en_\ell$ as in \eqref{eq:q} and \eqref{eq:E_l}. If each state ensemble $\en_\ell$ satisfies the hypothesis of Prop.~\ref{prop:easy} for all $\ell\in I$, then also $\en$ does it, and
\begin{equation*}
\Pg(\en) = \max_{\ell\in I} q(\ell) \Pg(\en_\ell) \,.
\end{equation*}
\end{corollary}
\begin{proof}
Using the notations of Prop.~\ref{prop:easy} for the ensemble $\en$, and denoting by $\lam_\ell(x)$ and $\Pi_\ell(x)$ the largest eigenvalue of $\en_\ell(x)$ and the corresponding eigenprojection, we have
$$
\lam(x) = q(\ell)\lam_\ell(x) \quad \text{and} \quad \Pi(x) = \Pi_\ell(x) \quad \text{for all } x\in X_\ell \,.
$$
Setting as usual
$$
\lam_{\en_\ell} = \max_{x\in X_\ell}\lam_\ell(x) \qquad X_{\en_\ell} = \{x\in X_\ell : \lam_\ell(x) = \lam_{\en_\ell}\} \,,
$$
the hypothesis is that
$$
\sum_{x\in X_{\en_\ell}} \Pi_\ell(x) = \mu_\ell \id \quad \text{for some $\mu_\ell\in\R$ and all $\ell\in I$} \,.
$$
Then,
\begin{align*}
& \lam_\en = \max_{\ell\in I} q(\ell) \lam_{\en_\ell} \\
& X_\en = \bigcup_{\ell\in I_0} X_{\en_\ell} \quad \text{where} \quad I_0 = \{\ell\in I : q(\ell) \lam_{\en_\ell} = \lam_\en\} \\
& \sum_{x\in X_\en} \Pi(x) = \sum_{\ell\in I_0} \, \sum_{x\in X_{\en_\ell}} \Pi_\ell(x) = \sum_{\ell\in I_0} \mu_\ell \id \,.
\end{align*}
Therefore, the state ensemble $\en$ satisfies condition \eqref{eq:easy_condition_1}. In paticular, by \eqref{item:(lambda)prop_easy} of Prop.~\ref{prop:easy},
$$
\Pg(\en) = d\lam_\en = d \max_{\ell\in I} q(\ell) \lam_{\en_\ell} = \max_{\ell\in I} q(\ell) \Pg(\en_\ell) \,.
$$
\end{proof}

A situation where Prop.~\ref{prop:easy} is applicable occurs, for instance, when a state ensemble $\en$ is invariant under an irreducible projective unitary representation of some symmetry group. 
More precisely, suppose $G$ is a finite group, and let $U$ be a projective unitary representation of $G$ on $\hi$. 
We say that a state ensemble $\en$ is $U$-\emph{invariant} if $U(g)\en(X)U(g)^* = \en(X)$ for all $g\in G$, where $\en(X) = \{\en(x) : x\in X\}$. 
The definition of $U$-invariance for a state ensemble was first given in \cite{Holevo73}, where an action of the group $G$ on the index set $X$ was also required; see also \cite{ElMeVe04}.
Further, we call a state ensemble $\en$ {\em injective} if it is injective as a function, i.e., $\en(x)\neq\en(y)$ for $x\neq y$.

\begin{proposition}\label{prop:cov}
Suppose the projective unitary representation $U$ is irreducible, and let $\en$ be an injective and $U$-invariant state ensemble. 
Then, condition \eqref{eq:easy_condition_1} holds for some $\mu\in\real$.
\end{proposition}

\begin{proof}
Since $\en$ is injective and $U$-invariant, we can define an action of $G$ on the index set $X$ by setting $g\cdot x=\en^{-1}\left(U(g)\en(x)U(g)^*\right)$ for all $g\in G$ and $x\in X$. 
Then $\en(g \cdot x) = U(g)\en(x)U(g)^*$. 
Hence, with the notations of Prop.~\ref{prop:easy}, we have $\lambda(g\cdot x) = \lambda(x)$ and $\Pi(g\cdot x) = U(g)\Pi(x)U(g)^*$. 
It follows that $g\cdot X_\en = X_\en$, and $U(g)\left(\sum_{x\in X_\en}\Pi(x)\right) = \left(\sum_{x\in X_\en}\Pi(x)\right)U(g)$. 
The irreducibility of $U$ then implies $\sum_{x\in X_\en}\Pi(x) = \mu\id$ for some $\mu\in\real$ by Schur Lemma.
\end{proof}

If $\en(X) = \{U(g)\en(x_0)U(g)^* : g\in G\}$ for some (hence for any) $x_0\in X$, and $U(g)\en(x_0)U(g)^* \neq \en(x_0)$ for all $g\in G\setminus\{1\}$, Props.~\ref{prop:easy} and \ref{prop:cov} is \cite[Theorem 4.2]{Holevo73}. 
Under the additional constraint $\tr{\en(x)} = \tr{\en(y)}$ for all $x,y\in X$, $U$-invariant state ensembles are named {\em compound geometrically uniform} (CGU) state sets in the terminology of \cite{ElMeVe04}.

%%%%%%%%%%%%%%%%%
\section{Qubit state ensembles with dihedral symmetry}\label{sec:qubit}
%%%%%%%%%%%%%%%%%

In this section we illustrate the previous general results with three examples of different qubit state ensembles. 
The first one (Sec.~\ref{sec:2_spin_eigenbases}) has been already treated in \cite{GoWe10} according to the approach explained in Sec.~\ref{sec:comparison}. 
We provide the solution also for that example, since our method further allows to establish when the problem has a unique optimal measurement. 

%%%%%%%%%%%%%%%%%%%%%%%
\subsection{Notation}
%%%%%%%%%%%%%%%%%%%%%%%

The Hilbert space of the system is $\hi=\complex^2$. We denote by $\vec{\sigma}=(\sigma_1,\sigma_2,\sigma_3)$ the vector of three Pauli matrices, and 
\begin{equation*}
\vec{v}\cdot \vec{\sigma} = v_1 \sigma_1 + v_2 \sigma_2 + v_2 \sigma_3
\end{equation*}
for all $\vec{v}\in\real^3$.  
For any nonzero vector $\vec{v}$, we write $\hat{v} = \vec{v}/\no{\vec{v}}$.
Further, we let $\hat{e}_1$, $\hat{e}_2$ and $\hat{e}_3$ be the unit vectors along the three fixed coordinate axes.

All of the three examples to be presented share a common symmetry group, i.e., the dihedral group $D_2$, consisting of the identity element $1$, together with the three $180^\circ$ rotations $\alpha$, $\beta$ and $\gamma$ along $\hat{e}_1$, $\hat{e}_2$ and $\hat{e}_3$, respectively. 
This group acts on $\complex^2$ by means of the projective unitary representation
$$
U(1) = \id \, , \quad U(\alpha) = \sigma_1 \, ,  \quad U(\beta) = \sigma_2 \,,  \quad U(\gamma) = \sigma_3 \,.
$$
The representation $U$ is irreducible as the operators $\{U(g): g\in D_2\}$ span the whole space $\ca{L}(\complex^2)$.

We will use the Bloch representation of qubit states; all states on $\complex^2$ are parametrized by vectors $\vec{a}\in\real^3$ with \(\no{\vec{a}}\leq 1\), the state corresponding to $\vec{a}$ being
\begin{equation*}
\varrho_{\vec{a}} = \frac{1}{2}(\id + \vec{a}\cdot\vec{\sigma}) \, .
\end{equation*}
For any nonzero vector $\vec{a}$, the eigenvalues $\lambda_+$, $\lambda_-$ of $\varrho_{\vec{a}}$ and the corresponding eigenprojections $\Pi_+$, $\Pi_-$ are
\begin{align*}
\lambda_\pm = \frac{1}{2} \left(1\pm\no{\vec{a}}\right) \, , \quad \Pi_\pm = \frac{1}{2}\left(\id \pm \vau\cdot\vsigma \right) \, .
\end{align*}

%%%%%%%%%%%%%%%%%%%%%%%%%%%%%%%%%%%%%%%%%%%%%
\subsection{Two equally probable qubit eigenbases}\label{sec:2_spin_eigenbases}
%%%%%%%%%%%%%%%%%%%%%%%%%%%%%%%%%%%%%%%%%%%%%

In the first example, the total state ensemble consists of four pure states $\varrho_{\pm\vau}$ and $\varrho_{\pm\vbu}$, where 
\begin{align*}
\vau & = \cos(\theta/2)\,\veu_1 + \sin(\theta/2)\,\veu_2 \\
\vbu & = \cos(\theta/2)\,\veu_1 - \sin(\theta/2)\,\veu_2 \, ,
\end{align*}
and $\theta$ is an angle in the interval $(0,\pi)$. 
The states $\varrho_{\pm\vau}$ and $\varrho_{\pm\vbu}$ are the eigenstates of the operators $\vau\cdot\vsigma$ and $\vbu\cdot\vsigma$, respectively.
The label set is chosen to be $X = \{+\vau,-\vau,+\vbu,-\vbu\}$.
We assume that all states are equally likely; thus the state ensemble $\en$ is
\begin{equation*}
\en(\vxu) = \frac{1}{8}\left(\id + \vxu\cdot\vec{\sigma}\right) \, ,  \qquad \vxu\in X \, .
\end{equation*}
We will then consider the partition $\Pscr = (X_a,X_b)$, with $X_\ell=\{+\hat{\ell},-\hat{\ell}\}$. As usual, the corresponding state subensembles are denoted by $\en_a$ and $\en_b$, and $q(\ell) = 1/2$ is the probability that a label occurs in $X_\ell$.

Since $\en_\ell(\pm\hat{\ell}) = \frac{1}{4}(\id \pm\hat{\ell}\cdot\vec{\sigma})$, we see that each state ensemble $\en_\ell$ corresponds to preparing one of the two orthogonal pure states $\varrho_{+\hat{\ell}}$, $\varrho_{-\hat{\ell}}$ with equal probability. So, the sharp measurements
\begin{equation*}
\N_{a}(\pm\vau) = \frac{1}{2}\big(\id \pm \vau\cdot\vec{\sigma}\big) \, , \quad \N_{b}(\pm\vbu) = \frac{1}{2}\big(\id \pm \vbu \cdot\vec{\sigma}\big) \, ,
\end{equation*}
perfectly discriminate $\en_{a}$ and $\en_{b}$, respectively; in particular, $\Pg(\en_a) = \Pg(\en_b) = 1$, hence $\Prg(\en;\Pscr)=1$. Moreover, Cor.~\ref{prop:easy-3} applies, and we conclude that $\Pg(\en) =1/2$.
The value of $\Pg(\en)$ can be obtained in various different ways, see e.g.~\cite{Bae13}.
Interestingly, $\Pg(\en)$ does not depend on the angle $\theta$.

To calculate the posterior information guessing probability $\Ppg(\en;\Pscr)$, we apply Prop.~\ref{prop:equiv} and calculate $\Pg(\enf)$, where the auxiliary state ensemble $\enf$ on $X_a\times X_b$ is given as
\begin{align*}
\enf(h\vau,k\vbu) & = \frac{1}{8} \left[ \id+\frac{1}{2}(h\vau+k\vbu)\cdot\vec{\sigma}\right]\,, \qquad h,k\in\{+,-\} \\
& =\begin{cases}
\frac{1}{4} \varrho_{h\cos(\theta/2)\,\veu_1} & \text{ if $h=k$} \\
\frac{1}{4} \varrho_{h\sin(\theta/2)\,\veu_2} & \text{ otherwise} \,,
\end{cases}
\end{align*}
and $\Delta=2$.
This state ensemble is clearly injective, and it is $U$-invariant as the set $\{h\vau+k\vbu : h,k=\pm\}\subset\real^3$ is invariant under the action of the dihedral group $D_2$. 
Hence, Prop.~\ref{prop:easy} is applicable by virtue of Prop.~\ref{prop:cov}, and it thus leads us to find the largest eigenvalue of $\enf(h\vau,k\vbu)$ and the corresponding eigenprojection. 
We obtain
\begin{align*}
\lambda(h\vau,k\vbu) & =
\begin{cases}
\frac{1}{8} \left(1+\sqrt{\frac{1+\cos\theta}{2}}\right) & \text{ if $h=k$} \\
\frac{1}{8} \left(1+\sqrt{\frac{1-\cos\theta}{2}}\right) & \text{ otherwise}
\end{cases} \\
\Pi(h\vau,k\vbu) & =
\begin{cases}
\frac{1}{2}(\id + h\sigma_1) & \text{ if $h=k$} \\
\frac{1}{2}(\id + h\sigma_2) & \text{ otherwise} \,,
\end{cases}
\end{align*}
and hence
\begin{equation}
\Ppg(\en;\Pscr) = \Delta\cdot d \cdot \lambda_\enf = \frac{1}{2} \left(1+\sqrt{\frac{1+\abs{\cos\theta}}{2}}\right) \,.
\end{equation}
As one could have expected, the unique minimum is in $\theta=\pi/2$ and the guessing probabilities are the same for $\theta_1$ and $\theta_2$ when $\theta_2 = \pi - \theta_1$; see Fig.~\ref{fig:2qubit}.

As shown in \cite{BuHe08}, we have $\jmd(\N_{a},\N_{b}) =\jmdu(\N_{a},\N_{b})= 1/\sqrt{1 + \mo{\sin\theta}}$. 
Therefore, the lower bound for $\Ppg(\en;\Pscr)$ given in \eqref{eq:jmd-bound} is
\begin{equation}\label{eq:jmd-bound-qubit}
\Ppg(\en;\Pscr) \geq \frac{1}{2} \left( 1 + \frac{1}{\sqrt{1+\mo{\sin\theta}}} \right) . 
\end{equation}
We see that the right hand side agrees with $\Ppg(\en;\Pscr)$ if and only if $\theta = \pi /2$; see Fig.~\ref{fig:2qubit}.
In particular, this shows that the noisy versions of the form \eqref{eq:noisy-uniform} are optimal only in the case $\theta = \pi /2$.

\begin{figure}
\includegraphics[width=7cm]{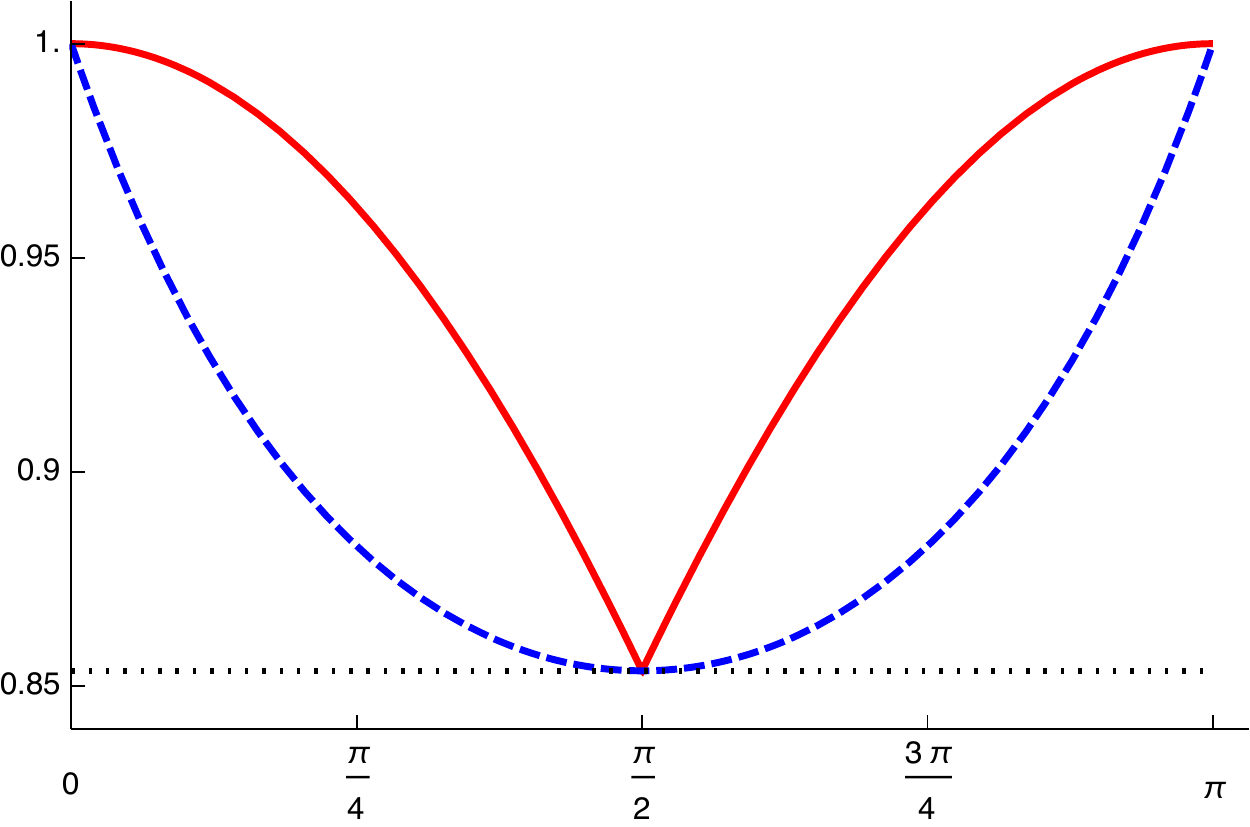}%(Color online)
\caption{\label{fig:2qubit} The red solid curve is the optimal guessing probability $\Ppg(\en;\Pscr)$ as a function of the angle $\theta$ between $\vau$ and $\vbu$, while the blue dashed curve is the lower bound \eqref{eq:jmd-bound-qubit} coming from the optimal joint measurement of uniform noisy versions of $\N_a$ and $\N_b$. }
\end{figure}

In order to find all optimal measurements, we distinguish the three cases $\theta\in (0,\pi/2)$, $\theta = \pi/2$ and $\theta\in (\pi/2,\pi)$.

\subsubsection{Case $\theta\in (0,\pi/2)$}\label{sec:<90}

We have $(X_a\times X_b)_\enf = \{(+\vau,+\vbu),(-\vau,-\vbu)\}$, and the projections $\{\Pi(h\vau,k\vbu) : (h\vau,k\vbu) \in (X_a\times X_b)_\enf\}$ are rank-$1$ linearly independent operators. 
From Cor.~\ref{prop:easy-2} we conclude that the measurement $\C_+$, defined as
\begin{equation*}
\C_+(h\vau,k\vbu) = \begin{cases}
\frac{1}{2} \left( \id + h\sigma_1 \right) & \text{ if $h = k$} \\
0 & \text{ otherwise}
\end{cases}
\end{equation*}
is the unique measurement on $X_a\times X_b$ achieving $\Pg(\enf)$, and hence also $\Ppg(\en;\Pscr)$. 
The two marginals of $\C_+$, $\pi_{1 \ast} \C_+$ and $\pi_{2 \ast} \C_+$, are such that $\pi_{1\ast} \C_+(\pm\vau) = \pi_{2\ast} \C_+(\pm\vbu)$, and Bob therefore is not using the post-measurement information to guess the spin value $+$ or $-$.
Bob can, in fact, choose a measurement $\M$ with outcomes $\{+,-\}$, $\M(h) = \frac{1}{2} \left( \id + h\sigma_1 \right)$, and when Alice announces that her choice was from subset $X_\ell$, Bob's guess is $h\hat{\ell}$, where $h$ is the outcome of $\M$.

In GW's approach, this is a situation in which post-measurement information is useless \cite[Subsec.~III\,C]{GoWe10}, as the diagonal measurement $\C_+'(h,k) = \delta_{h,k} \C_+(h\vau,k\vbu)$ is optimal for the task of discriminating a string in $\ca{X} = \{+,-\}$ even if Alice announces her encoding $a$ or $b$ after the measurement (see Sec.~\ref{sec:comparison}). In spite of this fact, $\Ppg(\en;\Pscr)$ is strictly larger than $\Pg(\en)$.
The reason is that with post-measurement information Bob gets the correct index $\ell\in\{a,b\}$ for free, so he can optimize his measurement to distinguish between one of the two alternatives $\pm\hat{\ell}$, instead of four alternatives $\pm\vau,\pm\vbu$. 

\subsubsection{Case $\theta\in (\pi/2,\pi)$}\label{sec:>90}

Now $(X_a\times X_b)_\enf = \{(+\vau,-\vbu),(-\vau,+\vbu)\}$; hence, proceeding as in the previous case, we find that the unique optimal measurement is
\begin{equation*}
\C_-(h\vau,k\vbu) = \begin{cases}
0 & \text{ if $h = k$} \\
\frac{1}{2} \left( \id + h\sigma_2 \right) & \text{ otherwise} \,.
\end{cases}
\end{equation*}
In this case, we have $\pi_{1 \ast} \C_-(\pm\vau) \neq \pi_{2 \ast} \C_-(\pm\vbu)$, therefore Bob is actually using post-measurement information (see \cite[Subsec.~III\,C]{GoWe10}).

\subsubsection{Case $\theta = \pi/2$}\label{sec:=90}

In this case, $(X_a\times X_b)_\enf = X_a\times X_b$, and the projections $\{\Pi(h\vau,k\vbu) : (h\vau,k\vbu) \in (X_a\times X_b)_\enf\}$ are not linearly independent. 
However, they are still rank-$1$, hence, by \eqref{item:(max)prop_easy} of Prop.~\ref{prop:easy}, any measurement maximizing $\Pg(\enf;\C)$ is of the form $\C(h\vau,k\vbu) = \alpha(h,k)\Pi(h\vau,k\vbu)$ for some function $\alpha : \{+,-\}^2 \to [0,1]$. 
The normalization condition $\sum_{h,k} \C(h\vau,k\vbu) = \id$ imposes 
\begin{align*}
&\alpha(+,+) = \alpha(-,-) = t \\
&\alpha(+,-) = \alpha(-,+) = 1-t
\end{align*}
for some $t\in [0,1]$. 
Therefore, an optimal measurement is any convex combination of the two measurements $\C_+$ and $\C_-$ found earlier. 
The convex combination $\C_0 = \half \C_+ + \half\C_-$ is the optimal measurement given in \eqref{eq:M0}, which in this case reads
\begin{equation*}
\C_0(h\vau,k\vbu) = \frac{1}{4} \left(\id + \frac{h+k}{2}\,\sigma_1 + \frac{h-k}{2}\,\sigma_2\right) \,.
\end{equation*}
Its marginals are
\begin{align*}
\pi_{1 \ast} \C_0(h\vau) & = \frac{1}{\sqrt{2}} \, \N_{a}(h\vau) + \left(1- \frac{1}{\sqrt{2}}\right) \frac{1}{2} \id \, , \\
\pi_{2 \ast} \C_0(k\vbu) & = \frac{1}{\sqrt{2}}\, \N_{b}(k\vbu) + \left(1- \frac{1}{\sqrt{2}}\right) \frac{1}{2} \id \, .
\end{align*}
These are noisy versions of the optimal measurements $\N_{a}$ and $\N_{b}$ for the maximization problems $\max_{\M} \Pg(\en_{a};\M)$ and $\max_{\M} \Pg(\en_{b};\M)$, respectively.
In this case, as we already observed, one implementation of the optimal startegy is hence to make an approximate joint measurement of $\N_{a}$ and $\N_{b}$. 

%%%%%%%%%%%%%%%%%%%%%%%%%%%%%%%%%%%%%%
\subsection{Two qubit state ensembles with dihedral $D_{2n}$-symmetry}
%%%%%%%%%%%%%%%%%%%%%%%%%%%%%%%%%%%%%%

Let us consider a state ensemble $\en$, labeled by the $2n+2$ labels $X=\{+,-,0,1,\ldots,2n-1\}$, and defined as
\begin{align*}
& \en(\pm) = \frac{q_1}{4} \left(\id \pm \sigma_1\right) \, , \\
& \en(k) = \frac{q_2}{4n} \left(\id + \vau_k \cdot\vec{\sigma}\right) \, , \quad k=0,\ldots,2n-1 \, , 
\end{align*}
where
\begin{equation*}
\vau_k =\cos\left(\pi k/n\right)\,\veu_2 + \sin\left(\pi k/n\right)\,\veu_3 
\end{equation*}
and $q_1,q_2>0$, $q_1+q_2=1$.

We consider the partition $\Pscr = (X_\ell)_{\ell\in\{1,2\}}$ of $X$, with $X_1 = \{+,-\}$ and $X_2 = \{0,1,\ldots,2n-1\}$. The corresponding subensembles are $\en_1(\pm) = \frac{1}{4} \left(\id \pm \sigma_1\right)$ and $\en_2(k) = \frac{1}{4n} \left(\id + \vau_k \cdot\vec{\sigma}\right)$, and the probability $q$ is $q(\ell) = q_\ell$.

Each of the two state ensembles $\en_1$, $\en_2$ is injective and $D_2$-invariant. By Prop.~\ref{prop:cov} (or even by direct inspection), it follows that $\en_1$ and $\en_2$ satisfy the hypothesis of Prop.~\ref{prop:easy}. 
Then, by Cor.~\ref{prop:easy-3} we have
\begin{align*}
\Pg(\en) & = \max\{q_1\Pg(\en_1),\,q_2\Pg(\en_2)\} \,.
\end{align*}
The subensemble $\en_1$ consists of two orthogonal pure states, hence it can be perfectly discriminated with the measurement $\N_1(\pm) =  \frac{1}{2} \left(\id \pm \sigma_1\right)$. On the other hand, an optimal measurement to discriminate the states in $\en_2$ is $\N_2(k) = \frac{1}{2n} \left(\id + \vau_k\cdot\vec{\sigma}\right)$ by \eqref{item:(max)prop_easy} of Prop.~\ref{prop:easy}, and $\Pg(\en_2) = 1/n$ by \eqref{item:(lambda)prop_easy} of the same proposition.
It follows that
\begin{align*}
\Pg(\en) & =\max\left\{q_1,\,\frac{q_2}{n}\right\} = \begin{cases}
q_1 & \text{ if } q_1 > \frac{1}{n+1} \\
\frac{1-q_1}{n} & \text{ if } q_1 \leq \frac{1}{n+1}
\end{cases}
\end{align*}
and
\begin{align*}
\Prg(\en;\Pscr)  = q_1 + \frac{q_2}{n} = \frac{(n-1)q_1 + 1}{n} \, .
\end{align*}

To calculate $\Ppg(\en;\Pscr)$, we first form the auxiliary state ensemble $\enf$ of Prop.~\ref{prop:equiv}. Its label set is the Cartesian product $\{+,-\}\times\{0,1,\ldots,2n-1\}$, it has $\Delta=2(nq_1+q_2)$ and it is given by
\begin{equation*}
\enf(h,k) = \frac{1}{8n} \left[ \id + \frac{1}{nq_1 + q_2} \left( nq_1 h\veu_1 + q_2\vau_k\right)\cdot\vec{\sigma} \right] \,.
\end{equation*}
The state ensemble $\enf$ is injective and $U$-invariant. 
Although the symmetry group of $\enf$ can be extended to the order $4n$ dihedral group $D_{2n}\supset D_2$, Prop.~\ref{prop:cov} yields that $D_2$-symmetry is already enough to ensure the applicability of Prop.~\ref{prop:easy}. 
The largest eigenvalue of $\enf(h,k)$ and the corresponding eigenprojection are found to be
\begin{align*}
\lambda(h,k) & = \frac{1}{8n} \left( 1 + \frac{\sqrt{n^2 q_1^2 + q_2^2}}{nq_1 + q_2} \right) \\
 \Pi(h,k) & = \frac{1}{2} \left( \id + \frac{nq_1 h\veu_1 + q_2\vau_k}{\sqrt{n^2 q_1^2 + q_2^2}}\cdot\vec{\sigma} \right) \,.
\end{align*}
It follows that $(X_1\times X_2)_\enf = X_1\times X_2$. The operators $\{\Pi(h,k) : (h,k)\in (X_1\times X_2)_\enf\}$ are rank-$1$ but they are not linearly independent.
Thus, we do not have uniqueness of optimal measurements.
By Theorem~\ref{prop:equiv} and Prop.~\ref{prop:easy}, 
\begin{align*}
& \Ppg(\en;\Pscr) = \Delta\cdot d\cdot\lam_\enf \\
& \quad = \frac{1}{2n} \left( nq_1 + q_2 + \sqrt{n^2 q_1^2 + q_2^2} \right) \\
& \quad = \frac{1}{2n} \left[ (n-1) q_1 + 1 + \sqrt{(n^2+1) q_1^2 - 2q_1 + 1} \right] \, .
\end{align*}
This maximum is attained by $\Ppg(\en;\Pscr;\C)$ if and only if $\C$ is a measurement of the form
\begin{equation*}
\C(h,k) = \frac{\alpha(h,k)}{2} \left[ \id + \frac{nq_1 h\veu_1 + q_2\vau_k}{\sqrt{n^2 q_1^2 + q_2^2}}\cdot\vec{\sigma} \right]
\end{equation*}
where $\alpha : X_1\times X_2 \to [0,1]$ is such that
\begin{align*}
& \sum_{k} \alpha(+,k) = \sum_{k} \alpha(-,k) = 1 \\
& \sum_{k} \left(\alpha(+,k) + \alpha(-,k)\right) \vau_k = 0
\end{align*}
by the normalization condition for $\C$. 

By choosing the constant function $\alpha(h,k) = \frac{1}{2n}$, we recover the optimal measurement of \eqref{eq:M0}.
The marginals of that measurement are the noisy versions of the measurements $\N_1(\pm) =  \frac{1}{2} \left(\id \pm \sigma_1\right)$ and $\N_2(k) = \frac{1}{2n} \left(\id + \vau_k\cdot\vec{\sigma}\right)$ optimally discriminating the subensembles $\en_1$ and $\en_2$.

%%%%%%%%%%%%%%%%%%%%%%%%%%
\subsection{Three orthogonal qubit eigenbases}
%%%%%%%%%%%%%%%%%%%%%%%%%%

Next we consider a state ensemble $\en$ with $6$ elements, having the index set
$X=\{+\veu_1,-\veu_1,+\veu_2,-\veu_2,+\veu_3,-\veu_3\}$ and defined as
\begin{equation*}
\en(\pm\veu_\ell) = \frac{q_\ell}{4} \left(\id \pm \sigma_\ell\right)\,, \qquad \ell\in\{1,2,3\} \, ,
\end{equation*}
where $q_1,q_2,q_3>0$ and $q_1+q_2+q_3=1$.
As the partition of $X$, we fix $\Pscr = (X_\ell)_{\ell\in\{1,2,3\}}$ with $X_\ell = \{+\veu_\ell,-\veu_\ell\}$. The corresponding subensembles are $\en_\ell(\pm\hat{\ell}) = \frac{1}{4} \left(\id \pm \sigma_\ell\right)$, and $q(\ell) = q_\ell$.

Each subsensemble $\en_\ell$ consists of orthogonal pure states and hence can be discriminated with the probability 1, the optimal measurement being $\N_\ell(\pm\hat{\ell}) = \frac{1}{2} ( \id \pm \sigma_\ell)$.
We thus have $\Prg (\en;\Pscr)=1$, and from Cor.~\ref{prop:easy-3} follows that $\Pg(\en) = \max\{q_1,q_2,q_3\}$.

To calculate $\Ppg (\en;\Pscr)$, we again form the auxiliary state ensemble $\enf$, which in this case is
\begin{equation*}
\enf(k_1\veu_1,k_2\veu_2,k_3\veu_3) = \frac{1}{16}\Bigg(\id + \sum_{\ell=1}^3 q_\ell k_\ell\sigma_\ell \Bigg) \,.
\end{equation*}
In the above formula, $k_\ell\in\{+,-\}$; moreover, we have $\Delta = 4$. As in the previous cases, the state ensemble $\enf$ is injective and $U$-invariant, and Prop.~\ref{prop:easy} then applies. 
We obtain
\begin{align*}
\lambda(k_1\veu_1,k_2\veu_2,k_3\veu_3) & = \frac{1}{16}\left( 1 + \no{\vec{q}}\right) \\
\Pi(k_1\veu_1,k_2\veu_2,k_3\veu_3) & = \frac{1}{2}\Bigg(\id + \frac{1}{\no{\vec{q}}} \sum_{\ell=1}^3 q_\ell k_\ell\sigma_\ell\Bigg) \,,
\end{align*}
where we set $\vec{q} = \sum_{\ell=1}^3 q_\ell \veu_\ell$.
Therefore, 
\begin{equation*}
\Ppg (\en;\Pscr) = \Delta\cdot d\cdot\lam_\enf = \frac{1}{2}\left( 1 + \sqrt{q_1^2 + q_2^2 + q_3^2}\right) \,.
\end{equation*}

In the case $q_1=q_2=q_3=1/3$, we have $\Ppg (\en;\Pscr) = (1+1/\sqrt{3})/2$.
As explained in Appendix \ref{app:two}, $\jmd(\N_1,\N_2,\N_3)=\jmdu(\N_1,\N_2,\N_3) = 1/\sqrt{3}$.
Therefore, the guessing probability with post-measurement information equals with the lower bound given in \eqref{eq:jmd-bound}, and we conclude that one way to implement the optimal measurement is to make a joint measurement of noisy versions of $\N_1,\N_2,\N_3$.

Since $(X_1\times X_2\times X_3)_\enf = X_1\times X_2\times X_3$ and all operators $\Pi(k_1\veu_1,k_2\veu_2,k_3\veu_3)$'s are rank-$1$, any optimal measurement is of the form
\begin{align*}
& \C(k_1\veu_1,k_2\veu_2,k_3\veu_3) = \frac{\alpha(k_1,k_2,k_3)}{2}\Bigg(\id + \frac{1}{\no{\vec{q}}} \sum_{\ell=1}^3 q_\ell k_\ell\sigma_\ell\Bigg)
\end{align*}
for some function $\alpha:\{+,-\}^3 \to [0,1]$. 
The normalization of $\C$ implies that, for every $k_1,k_2,k_3 \in \{+,-\}$, 
\begin{equation*}
\sum_{i,j}\alpha(k_1,i,j) = \sum_{i,j}\alpha(i,k_2,j) = \sum_{i,j}\alpha(i,j,k_3) = 1 \, .
\end{equation*}
One solution is to take the constant function $\alpha \equiv 1/4$, and that choice gives the optimal measurement of \eqref{eq:M0}. 
The marginals of this measurement are noisy versions of $\N_1,\N_2$ and $\N_3$. Another possibility is
$$
\alpha(k_1,k_2,k_3) = \begin{cases}
1 & \text{ if } k_1 = k_2 = k_3 \\
0 & \text{ otherwise}\,.
\end{cases}
$$
In GW's approach of Sec.~\ref{sec:comparison}, the latter choice corresponds to the diagonal optimal measurement $$
\C'(k_1,k_2,k_3) = \begin{cases}
\frac{1}{2} \left(\id + k \hat{q}\cdot\vsigma\right) & \text{ if } k_1 = k_2 = k_3 \equiv k \\
0 & \text{ otherwise}\,.
\end{cases}
$$
In particular, we see that from the point of view of GW's approach, post-measurement information is useless in this example.

%%%%%%%%%%%%%%%%%%%%%%%%%%%%%
\section{Two Fourier conjugate mutually unbiased bases}\label{sec:mub}
%%%%%%%%%%%%%%%%%%%%%%%%%%%%%

In this section, we consider the discrimination problem with post-measurement information for two mutually unbiased bases (MUB) in arbitrary finite dimension $d$. We restrict to the case in which the two bases are conjugated by the Fourier transform of the cyclic group $\Zb_d = \{0,\ldots,d-1\}$, endowed with the composition law given by addition ${\rm mod}\,d$. Moreover, we assume that all elements of each basis have equal apriori probabilities. However, we allow the occurrence probability of a basis to differ from that of the other one.

In formulas, we fix two orthonormal bases $(\varphi_h)_{h\in\Zb_d}$ and $(\psi_k)_{k\in\Zb_d}$ of $\hi$, such that
$$
\psi_k = \frac{1}{\sqrt{d}} \sum_{h\in\Zb_d} \omega^{hk} \phii_h \qquad \text{where} \qquad \omega = \e^{\frac{2\pi\i}{d}} \,.
$$
They satisfy the mutual unbiasedness condition
\begin{equation*}
\abs{\ip{\varphi_h}{\psi_k}} = \frac{1}{\sqrt{d}} \qquad \forall h,k\in\Zb_d \,.
\end{equation*}

We label the two bases by means of the symbols $X_\phii=\{0\phii,\ldots,(d-1)\phii\}$ and $X_\psi=\{0\psi,\ldots,(d-1)\psi\}$, respectively, and we let $X=X_\phii \cup X_\psi$ be the overall label set. 
Notice that, consistently with the previous examples, the elements of \(X\) are denoted by juxtaposing the index of the vector with the symbol of the basis which the vector belongs to (for example, the symbol \(0\phii\) labels the first vector in the basis \((\varphi_h)_{h\in\Zb_d}\)).
Then, we partition $X$ and use it to construct a state ensemble $\en$ as follows:
\begin{align}
\Pscr & = (X_\ell)_{\ell\in\{\phii,\psi\}} \label{eq:MUB_partition} \\
\en(h\ell) & = \frac{q_\ell}{d} \kb{\ell_h}{\ell_h}\,, \qquad  h\ell\in X \,, \label{eq:MUB_ensemble}
\end{align}
where $q_\phii,q_\psi > 0$ with $q_\phii+q_\psi=1$. The partition $\Pscr$ yields the two subensembles
$$
\en_\ell(h\ell) = \frac{1}{d} \kb{\ell_h}{\ell_h} \,, \qquad h\in\Zb_d \,,
$$
with $\ell\in\{\phii,\psi\}$; the probability that a label occurs in the subset $X_\ell$ is $q(\ell) = q_\ell$.

Note that in Sec.~\ref{sec:2_spin_eigenbases}, the two equally probable qubit eigenbases with angle $\theta = \pi/2$ constitute two MUB that are conjugated by the Fourier transform of the cyclic group $\Zb_2$. Indeed, this follows by setting
\begin{align*}
\phii_0 & = \frac{1}{\sqrt{2}}\left(\e^{-\i\frac{\pi}{8}}\eta_1 + \e^{\i\frac{\pi}{8}}\eta_2\right) \\
\phii_1 & = \frac{\i}{\sqrt{2}}\left(\e^{-\i\frac{\pi}{8}}\eta_1 - \e^{\i\frac{\pi}{8}}\eta_2\right) \,,
\end{align*}
where $(\eta_1,\eta_2)$ is the canonical (computational) basis of $\complex^2$, choosing $q_\phii = q_\psi = 1/2$, and relabeling
$$
0\phii\to +\hat{a} \qquad 1\phii\to -\hat{a} \qquad 0\psi\to +\hat{b} \qquad 1\psi\to -\hat{b} \,.
$$

We define two measurements $\N_\phii$ and $\N_\psi$ with outcomes in $X_\phii$ and $X_\psi$, respectively, as
$$
\N_\ell(h\ell) = \kb{\ell_h}{\ell_h}  \qquad h\in\Zb_d \,.
$$
Each of these measurements perfectly discriminates the corresponding subensemble $\en_\ell$. Moreover, once again Cor.~\ref{prop:easy-3} can be applied, thus leading to
\begin{gather*}
\Pg(\en) = \max\{q_\phii,\,q_\psi\} = \abs{q_\phii - \frac{1}{2}} + \frac{1}{2} \\
\Prg(\en;\Pscr) = 1 \,.
\end{gather*}

By Theorem \ref{prop:equiv}, optimizing the posterior information guessing probability $\Ppg(\en;\Pscr;\C)$ over all measurements $\C$ on $X_\phii\times X_\psi$ amounts to the same optimization problem for $\Pg(\enf;\C)$, where $\enf$ is the auxiliary state ensemble
\begin{equation*}
\enf(h\phii,k\psi) = \frac{1}{d^2} \left( q_\phii \kb{\varphi_h}{\varphi_h} + q_\psi \kb{\psi_k}{\psi_k}\right) \, .
\end{equation*}
The state ensemble $\enf$ has the direct product abelian group $G=\Zb_d\times\Zb_d$ as its natural symmetry group. Indeed, by defining the generalized Pauli operators
$$
W(r,s) = \sum_{z\in\Zb_d} \omega^{sz} \kb{\phii_{r+z}}{\phii_z} \,,
$$
we obtain a projective unitary representation of $\Zb_d\times\Zb_d$, such that
\begin{gather*}
W(r_1,s_1)W(r_2,s_2) = \omega^{s_1 r_2} W(r_1+r_2,s_1+s_2) \\
W(r,s)\phii_h = \omega^{sh} \phii_{r+h} \qquad W(r,s)\psi_k = \omega^{-r(s+k)} \psi_{s+k}
\end{gather*}
(see e.g.~\cite{Schwinger60,CaHeTo12}; here, $W(r,s) = U_r V_s$ in terms of the discrete position and momentum displacement operators $U_r$ and $V_s$ defined in \cite[Subsec.~IV A]{CaHeTo12}). Then, the state ensemble $\enf$ is $W$-invariant, as
\begin{equation}\label{eq:covF}
W(r,s)\enf(h\phii,k\psi)W(r,s)^* = \enf((h+r)\phii,(k+s)\psi)
\end{equation}
for all $h,k,r,s\in\Zb_d$. Since the representation $W$ is irreducible \cite{Schwinger60} and the state ensemble $\enf$ is clearly injective, Prop.~\ref{prop:easy} can be applied to $\enf$ by Prop.~\ref{prop:cov}. In order to proceed as usual, we need the next lemma.

\begin{lemma}\label{lem:another_one}
For all $h,k\in\Zb_d$, the largest eigenvalue and the corresponding eigenprojection of $\enf(h\phii,k\psi)$ are
\begin{subequations}\label{eq:diagF}
\begin{align}
& \lambda(h\phii,k\psi) = \frac{1}{2d^2}\left[1+\sqrt{(q_\phii - q_\psi)^2 + \frac{4}{d}\,q_\phii q_\psi}\right] \label{eq:diagF_1}\\
\label{eq:diagF_2}
\begin{split}
& \Pi(h\phii,k\psi)
= \kb{\alpha\phii_h+\beta\omega^{-hk}\psi_k}{\alpha\phii_h+\beta\omega^{-hk}\psi_k} \\
& \quad = W(h,k) \kb{\alpha\phii_0+\beta\psi_0}{\alpha\phii_0+\beta\psi_0} W(h,k)^* \,,
\end{split}
\end{align}
\end{subequations}
where the couple $(\alpha,\beta)$ is the unique solution to the following system of equations:
\begin{subequations}\label{eq:sys_ab}
\begin{gather}
\alpha > 0, \quad \beta > 0 \label{eq:sys_ab_1} \\
\alpha^2 + \beta^2 + \frac{2}{\sqrt{d}} \alpha\beta = 1 \label{eq:sys_ab_2} \\
\frac{\alpha}{\beta} = \frac{\sqrt{d}}{2q_\psi} \left[q_\phii - q_\psi + \sqrt{(q_\phii - q_\psi)^2 + \frac{4}{d}\,q_\phii q_\psi}\right] \,. \label{eq:sys_ab_3}
\end{gather}
\end{subequations}
\end{lemma}
Eq.~\eqref{eq:sys_ab_2} describes an ellipse in the $\alpha\beta$-plane centered at $(0,0)$ and having the minor axis along the $\alpha=\beta$ direction. The solution of \eqref{eq:sys_ab} is where this ellipse intersects the half-line originating at $(0,0)$, lying in the first quadrant \eqref{eq:sys_ab_1} and having the positive slope given by \eqref{eq:sys_ab_3}.
\begin{proof}
By means of the covariance condition \eqref{eq:covF} for $\enf$, it is enough to prove \eqref{eq:diagF} only for $h=k=0$. In order to do it, we preliminarly observe that the operator $\enf(0\phii,0\psi)$ leaves the linear subspace $\hi_0 = \spanno{\phii_0,\psi_0}$ invariant, and it is null on $\hi_0^\perp$. Moreover, with respect to the linear (nonorthogonal) basis $(\phii_0,\psi_0)$ of $\hi_0$, the restriction of $\enf(0\phii,0\psi)$ to $\hi_0$ has the matrix form
$$
\left.\enf(0\phii,0\psi)\right|_{\hi_0} = \frac{1}{d^2} \left(\begin{array}{cc}
q_\phii & \frac{1}{\sqrt{d}}q_\phii \\
\frac{1}{\sqrt{d}}q_\psi & q_\psi
\end{array}\right) \,.
$$
The roots of the characteristic polynomial of the above matrix are
$$
\lam_{\pm} = \frac{1}{2d^2}\left[1\pm\sqrt{(q_\phii - q_\psi)^2 + \frac{4}{d}\,q_\phii q_\psi}\right]
$$
(recall $q_\phii + q_\psi = 1$), and they are clearly different. This gives \eqref{eq:diagF}. By direct inspection of the previous matrix, the vector $\chi=\alpha\phii_0 + \beta\psi_0$ is a nonzero $\lambda_+$-eigenvector of $\enf(0\phii,0\psi)$ if and only if the ratio $\alpha/\beta$ is given by \eqref{eq:sys_ab_3}. Normalization of $\chi$ gives \eqref{eq:sys_ab_2}. Since the ratio $\alpha/\beta$ is real and positive, \eqref{eq:sys_ab_2} and \eqref{eq:sys_ab_3} have a unique common solution satisfying \eqref{eq:sys_ab_1}.
\end{proof}

\begin{proposition}\label{prop:main_MUB}
For the state ensemble $\en$ of \eqref{eq:MUB_ensemble} and the partition $\Pscr$ of \eqref{eq:MUB_partition}, we have
\begin{equation}\label{eq:PpostMUB}
\Ppg(\en;\Pscr) = \frac{1}{2}\left[1+\sqrt{(q_\phii - q_\psi)^2 + \frac{4}{d}\,q_\phii q_\psi}\right] \,.
\end{equation}
Moreover, a measurement on $X_\phii\times X_\psi$ maximizing the guessing probability $\Ppg(\en;\Pscr;\C)$ is
\begin{equation}\label{eq:optiMUB}
\begin{aligned}
& \C_0(h\phii,k\psi) = \frac{1}{d} \, \kb{\alpha\phii_h+\beta\omega^{-hk}\psi_k}{\alpha\phii_h+\beta\omega^{-hk}\psi_k} \\
& \qquad = \frac{1}{d} W(h,k) \kb{\alpha\phii_0+\beta\psi_0}{\alpha\phii_0+\beta\psi_0} W(h,k)^* \,,
\end{aligned}
\end{equation}
where \((\alpha,\beta)\) is the solution to the system of equations \eqref{eq:sys_ab}.
The measurement $\C_0$ is the unique measurement maximizing the guessing probability $\Ppg(\en;\Pscr;\C)$ if and only if the dimension $d$ of $\hi$ is odd.
\end{proposition}

\begin{proof}
We have already seen that Prop.~\ref{prop:easy} can be applied to the state ensemble $\enf$. With the notations of that proposition, we have
\begin{gather*}
\lam_\enf = \frac{1}{2d^2}\left[1+\sqrt{(q_\phii - q_\psi)^2 + \frac{4}{d}\,q_\phii q_\psi}\right] \\
(X_\phii\times X_\psi)_\enf = X_\phii\times X_\psi
\end{gather*}
by Lemma \ref{lem:another_one}. In particular, the value of $\lam_\enf$ and Theorem \ref{prop:equiv} with $\Delta = d$ imply \eqref{eq:PpostMUB}. Moreover, still by Lemma \ref{lem:another_one}, the measurement $\C_0$ in \eqref{eq:optiMUB} is the optimal measurement \eqref{eq:M0} for the guessing probability $\Pg(\enf;\C)$, hence also for $\Ppg(\enf;\Pscr;\C)$. By Cor.~\ref{prop:easy-2}, there is no other measurement maximizing the guessing probability $\Ppg(\en;\Pscr;\C)$ if and only if the operators $\{\Pi(h\phii,k\psi) : h,k\in\Zb_d\}$ are linearly independent. The argument used in the proof of \cite[Prop.~9]{CaHeTo12} shows that this is equivalent to the dimension $d$ of $\hi$ being odd.
\end{proof}

In the particular case \(q_\phii = q_\psi=\frac{1}{2}\), formulas \eqref{eq:PpostMUB} and \eqref{eq:optiMUB} and  simplify as follows:
\begin{align*}
\Ppg(\en;\Pscr) & = \frac{1}{2}\left(1+\frac{1}{\sqrt{d}}\right) \\
\C_0(h\phii,k\psi) & = \frac{1}{2(\sqrt{d}+d)} \\
& \qquad \times W(h,k) \kb{\phii_0+\psi_0}{\phii_0+\psi_0} W(h,k)^* \,,
\end{align*}
which, for \(d=2\), are easily seen to be consistent with the results of Sec.~\ref{sec:2_spin_eigenbases}.

For general $q_\phii$, $q_\psi$, the first marginal of $\C_0$ is
\begin{align}
\pi_{1 \ast} \C_0(h\phii) = & \frac{1}{d} \sum_{k\in\Zb_d} \left[\alpha^2\kb{\phii_h}{\phii_h} + \beta^2 \kb{\psi_k}{\psi_k}\right. \nonumber \\
& \left. + \alpha\beta \left(\kb{\omega^{-hk}\psi_k}{\phii_h} + \kb{\phii_h}{\omega^{-hk}\psi_k}\right)\right] \nonumber \\
= & \alpha^2 \kb{\phii_h}{\phii_h} + \frac{\beta^2}{d} \id + \frac{2\alpha\beta}{\sqrt{d}} \kb{\phii_h}{\phii_h} \nonumber \\
= & t_\phii \N_\phii(h\phii) + (1-t_\phii)\frac{1}{d}\,\id \,, \label{eq:noisyQ}
\end{align}
where
$
t_\phii = \alpha^2 + \frac{2\alpha\beta}{\sqrt{d}} \,.
$
Here, we have used the fact that
$$
\sum_{k\in\Z_d} \omega^{-hk}\psi_k = \sqrt{d} \phii_h \qquad \forall h\in\Z_d \,.
$$
With a similar calculation,
\begin{align}
\pi_2 \ast \C_0(k\psi) & = t_\psi \N_\psi(h\psi) + (1-t_\psi)\frac{1}{d}\,\id \label{eq:noisyP} 
\end{align}
where
$
t_\psi = \beta^2 + \frac{2\alpha\beta}{\sqrt{d}} \,.
$
We conclude that the marginals of $\C_0$ are noisy versions of $\N_\phii$ and $\N_\psi$.

We remark that approximate joint measurements of $\N_\phii$ and $\N_\psi$ were studied in \cite{CaHeTo12}. In particular, by \cite[Props.~5 and 6]{CaHeTo12}, noisy measurements of the form \eqref{eq:noisyQ} and \eqref{eq:noisyP} are jointly measurable if and only if
\begin{subequations}
\begin{gather}
t_\phii + t_\psi \leq 1 \label{eq:triangle}\\
\text{or} \qquad t_\phii^2 + t_\psi^2 + \frac{2(d-2)}{d} (1-t_\phii)(1-t_\psi) \leq 1 \,. \label{eq:half-ellipse} 
\end{gather}
\end{subequations}
Moreover, regardless of the dimension $d$ of $\hi$, there is a unique joint measurement when the equality is attained in \eqref{eq:half-ellipse}.
One can confirm that $t_\phii$ and $t_\psi$ with $\alpha$ and $\beta$ given by \eqref{eq:sys_ab_2} lead to equality in \eqref{eq:half-ellipse}, hence $\C_0$ can be identified as that unique joint measurement. It also follows by Prop.~\ref{prop:main_MUB} that for even dimensions $d$ there are measurements maximizing $\Ppg(\en;\Pscr;\C)$ whose marginals $\pi_{1\ast}\C$ and $\pi_{2\ast}\C$ are not noisy versions of $\N_\phii$ and $\N_\psi$.

%%%%%%%%%%%%%%%%%%%%%%%%%%%%%%%%%%
\section{Acknowledgement}
%%%%%%%%%%%%%%%%%%%%%%%%%%%%%%%%%%

This work was performed as part of the Academy of Finland Centre of Excellence program (project 312058).

%%%%%%%%%%%%%%%%%%%%%%%%%%%%%%%%%%%

\newpage

%%%%%%%%%%%%%%%%%%%%%%%%%%%%%%%%%%
\begin{appendix}
%%%%%%%%%%%%%%%%%%%%%%%%%%%%%%%%%%

%%%%%%%%%%%%%%%%%%%%%%%%%%%%%%%%%%
\section{Necessity of the rank-$1$ condition in Corollary \ref{prop:easy-2}}\label{app:one}
%%%%%%%%%%%%%%%%%%%%%%%%%%%%%%%%%%

The following state ensemble $\en : \{1,2,3\}\to\elle{\complex^2}$
\begin{gather*}
\en(1) = \left(\begin{array}{ccc}
1/2 & 0 & 0 \\ 0 & 1/2 & 0 \\ 0 & 0 & 0
\end{array}\right) \qquad \en(2) = \left(\begin{array}{ccc}
1/2 & 0 & 0 \\ 0 & 0 & 0 \\ 0 & 0 & 1/2
\end{array}\right) \\
\en(3) = \left(\begin{array}{ccc}
0 & 0 & 0 \\ 0 & 1/2 & 0 \\ 0 & 0 & 1/2
\end{array}\right)
\end{gather*}
satisfies \eqref{eq:easy_condition_1} and item \eqref{it:(i)_prop:easy-2} of Cor.~\ref{prop:easy-2}. However, it does not fulfill item \eqref{it:(ii)_prop:easy-2} of the same corollary, as both the following measurements
$$
\M_0(x) = \en(x), \qquad x\in\{1,2,3\}\,,
$$
and
\begin{gather*}
\M_1(1) = \en(1) + \en(2) - \en(3) \\
\M_1(2) = \en(1) - \en(2) + \en(3) \\
\M_1(3) = -\en(1) + \en(2) + \en(3)
\end{gather*}
attain the maximum guessing probability $\Pg(\en;\M_i) = \Pg(\en) = 3/2$.

\newpage

%%%%%%%%%%%%%%%%%%%%%%%%%%%%%%%%%%
\section{Joint measurability degree of three orthogonal qubit measurements}\label{app:two}
%%%%%%%%%%%%%%%%%%%%%%%%%%%%%%%%%%

Let $\N_\ell(\pm\veu_\ell) = \half ( \id \pm \sigma_\ell)$ for $\ell=1,2,3$. 
We aim to show that $\jmd(\N_1,\N_2,\N_3)=1/\sqrt{3}$, which means that we need to find the largest $t$ such that the noisy versions 
\begin{equation}\label{eq:3noisy}
\nN_\ell(\pm\veu_\ell) = t \N_\ell(\pm\veu_\ell) + (1-t) \nu_\ell(\pm\veu_\ell) \id
\end{equation}
are compatible.
The probability distributions $\nu_1,\nu_2$ and $\nu_3$ can be chosen freely, meaning that we optimize among all their possible choices. 
It has been shown in \cite{LiSpWi11} that \(\jmdu(\N_1,\ldots,\N_m)=1/\sqrt{3}\).
Therefore, the remaining point in order to conclude that $\jmd(\N_1,\N_2,\N_3)=1/\sqrt{3}$ is  provided by the following result.

\begin{proposition}
If $\nN_1,\nN_2,\nN_3$ given by \eqref{eq:3noisy} are compatible, then the observables
\[
\nN'_\ell(\pm\veu_\ell) = t \N_\ell(\pm\veu_\ell) + (1-t) \frac12 \id \,, \qquad \ell=1,2,3\,,
\]
are also compatible. 
\end{proposition}

\begin{proof}
We assume that $\nN_1,\nN_2,\nN_3$ defined in \eqref{eq:3noisy} are compatible, and we let $\C$ be any measurement having marginals $\nN_1,\nN_2,\nN_3$.
We denote by $A:\complex^2 \to \complex^2$ the antiunitary operator satisfying $A\sigma_\ell A^* = - \sigma_\ell$ for $\ell=1,2,3$. 
Explicitly, \(A=\sigma_2 \, J\), where \(J\) denotes  complex conjugation with respect to the canonical basis of $\complex^2$.
We then define $\C'$ as
\begin{align*}
\C'(k_1 \veu_1,k_2 \veu_2,k_3 \veu_3) & = \half \left[ \C(k_1 \veu_1,k_2 \veu_2,k_3 \veu_3) \right.\\ &\left.+ A\C(-k_1 \veu_1,-k_2 \veu_2,-k_3 \veu_3)A^* \right] \, .
\end{align*}
A direct calculation shows that the marginals of $\C'$ are $\nN_1'$, $\nN_2'$, $\nN_3'$.
\end{proof}

%%%%%%%%%%%%%%%%%%%%%%%%%%%%%%%%%%
\end{appendix}
%%%%%%%%%%%%%%%%%%%%%%%%%%%%%%%%%%

%%%%%%%%%
%%%%%%%%%

\begin{thebibliography}{10}
\bibitem{Plavala16}
M.~Pl\'avala.
\newblock All measurements in a probabilistic theory are compatible if and only
  if the state space is a simplex.
\newblock {\em Phys. Rev. A}, 94:042108, 2016.
\end{thebibliography}
\end{document}